\documentclass[11pt]{article}
\usepackage{style}
\usepackage{times}
\usepackage{fullpage}

\newcommand{\seq} {dominance sequence}
\newcommand{\aseq} {approximate dominance sequence}
\newcommand{\Seq} {Dominance Sequence}
\newcommand{\pars} {partition sequence}
\newcommand{\cps} {compressed partition sequence}

\newcommand{\struct} {leveled bucket structure}
\newcommand{\bucket} {leveled buckets}
\newcommand{\stt} {\mb{start}}
\newcommand{\en} {\mb{end}}
\newcommand{\pw} {\mb{lg}}

\newcommand{\disPreserve} {distance preserving}

\newcommand{\myparagraph}[1]{\medskip\noindent {\bf #1.}}

\title{Efficient Construction of Probabilistic Tree Embeddings}
\author{Guy Blelloch\\Carnegie Mellon University\\guyb@cs.cmu.edu \and Yan Gu\\Carnegie Mellon University\\yan.gu@cs.cmu.edu \and Yihan Sun\\Carnegie Mellon University\\yihans@cs.cmu.edu}
\date{}
\begin{document}
    \maketitle
    \thispagestyle{empty}
    \par
\setcounter{secnumdepth}{2}
\setcounter{tocdepth}{2}

\begin{abstract}
In this paper we describe an algorithm that embeds a graph metric $(V,d_G)$ on an undirected weighted graph $G=(V,E)$
into a distribution of tree metrics $(T,D_T)$ such that for every pair $u,v\in V$,
$d_G(u,v)\leq d_T(u,v)$ and $\E_{T}[d_T(u,v)]\leq O(\log n)\cdot d_G(u,v)$.
Such embeddings have proved highly
useful in designing fast approximation algorithms, as many hard problems on graphs are easy to solve on
tree instances.
For a graph with $n$ vertices and $m$ edges, our algorithm runs in $O(m\log n)$ time with high probability, which improves the previous upper bound of $O(m\log^3 n)$ shown by Mendel et al.\,in 2009.

The key component of our algorithm is a new approximate single-source shortest-path algorithm, which implements the priority queue with a new data structure, the \emph{bucket-tree structure}.
The algorithm has three properties: it only requires linear time in the number of edges in the input graph;
the computed distances have a distance preserving property;
and when computing the shortest-paths to the $k$-nearest vertices from the source, it only requires to visit these vertices and their edge lists.
These properties are essential to guarantee the correctness and the stated time bound.

Using this shortest-path algorithm, we show how to generate an intermediate structure, the approximate dominance sequences of the input graph, in $O(m \log n)$ time, and further propose a simple yet efficient algorithm to converted this sequence to a tree embedding in $O(n\log n)$ time, both with high probability.
Combining the three subroutines gives the stated time bound of the algorithm.

Then we show that this efficient construction can facilitate some applications.  We proved that FRT trees (the generated tree embedding) are Ramsey partitions with asymptotically tight bound, so the construction of a series of distance oracles can be accelerated.
\end{abstract}

\newpage
\setcounter{page}{1}
\section{Introduction}

The idea of probabilistic tree embeddings~\cite{Bartal98} is to embed a
finite metric into a distribution of tree metrics with a minimum
expected distance distortion. A distribution $\mathcal{D}$ of trees of
a metric space $(X,d_X)$ should minimize the expected stretch $\psi$
so that:
\begin{enumerate}
  \item dominating property:  for each tree $T\in \mathcal{D}$, $d_X(x,y)\leq d_T(x,y)$ for every $x,y\in X$, and
  \item expected stretch bound: $\E_{T\sim\mathcal{D}}[d_T(x,y)]\leq \psi\cdot d_X(x,y)$ for every $x,y\in X$,
\end{enumerate}
where $d_T(\cdot,\cdot)$ is the tree metric, and $\E_{T\sim\mathcal{D}}$ draws
a tree $T$ from the distribution $\mathcal{D}$.
After a sequence of results~\cite{alon1995graph,Bartal96,Bartal98},
Fakcharoenphol, Rao and Talwar~\cite{FRT} eventually proposed an elegant and asymptotically optimal
algorithm (FRT-embedding) with $\psi=O(\log n)$.

Probabilistic tree embeddings facilitate many applications.
They lead to practical algorithms to solve a number of
problems with good approximation bounds, for example, the $k$-median
problem, buy-at-bulk network design~\cite{Kanat}, and network
congestion minimization~\cite{Racke08}.
A number of network algorithms use tree embeddings as key components,
and such applications include generalized Steiner forest problem, the minimum routing cost spanning tree problem, and the $k$-source shortest paths problem~\cite{Khan2008}.
Also, tree embeddings are
used in solving symmetric diagonally dominant (SDD) linear
systems.  Classic solutions use spanning trees as the preconditioner, but recent work by Cohen et al.~\cite{Cohen2014} describes a new approach to use trees with Steiner nodes (e.g.\ FRT trees).

In this paper we discuss yet another remarkable application of probabilistic tree embeddings:
constructing of approximate distance oracles (ADOs)---a data structure with
compact storage ($o(n^2)$) which can approximately and efficiently
answer pairwise distance queries on a metric space.
We show that FRT trees can be used to accelerate the construction
of some ADOs~\cite{Ramsey,WN,SC}.

Motivated by these applications, efficient algorithms to construct
tree embeddings are essential, and there are several results on the topic
in recent years~\cite{cohen1997,Khan2008,Kanat,FastCKR,ghaffari2014near,Friedrichs2016,blelloch2016parallelism}.
Some of these algorithms are based on different parallel settings, e.g.\ share-memory setting~\cite{Kanat,Friedrichs2016,blelloch2016parallelism} or distributed setting~\cite{Khan2008,ghaffari2014near}.
As with this paper, most of these algorithms~\cite{cohen1997,Khan2008,FastCKR,ghaffari2014near,blelloch2016parallelism} focus on graph
metrics, which most of the applications discussed above are based on.
In the sequential setting, i.e.\ on a RAM model, to the best of our knowledge, the most efficient algorithm to construct optimal
FRT-embeddings was proposed by Mendel and
Schwob~\cite{FastCKR}.  It constructs FRT-embeddings in
$O(m\log^3 n)$ expected time given an undirected positively weighted
graph with $n$ vertices and $m$ edges.
This algorithm, as well as the original construction in the FRT paper~\cite{FRT},
works hierarchically by generating each level of a tree top-down.  However,
such a method can be expensive in time and/or coding
complexity.  The reason is that the diameter of the graph can be
arbitrarily large and the FRT trees may contain many levels, which
requires complicated techniques, such as building
sub-trees based on quotient graphs.

\textbf{Our results.}
The main contribution of this paper is an efficient construction of the
FRT-embeddings.  Given an undirected positively weighted graph
$G=(V,E)$ with $n$ vertices and $m$ edges, our algorithm builds an optimal tree embedding in $O(m\log n)$ time.
In our algorithm, instead of generating partitions by level, we adopt
an alternative view of the FRT algorithm in~\cite{Khan2008,Kanat}, which computes the potential ancestors for each vertex using \textbf{\emph{\seq{}s}} of a graph (first proposed in~\cite{cohen1997}, and named as least-element lists in~\cite{cohen1997,Khan2008}).
The original algorithm to compute the \seq{}s requires $O(m\log n+n\log^2 n)$ time~\cite{cohen1997}.
We then discuss a new yet simple algorithm to convert the \seq{}s to an FRT tree
only using $O(n\log n)$ time.
A similar approach was taken by Khan et al.~\cite{Khan2008} but their output is an implicit representation (instead of an tree) and under the distributed setting and it is not work-efficient without using the observations and tree representations introduced by Blelloch et al.\ in~\cite{Kanat}.\footnote{A simultaneous work by Friedrichs et al.\,proposed an $O(n \log^3 n)$ algorithm of this conversion (Lemma 7.2 in~\cite{Friedrichs2016}). }

Based on the algorithm to efficiently convert the \seq{}s to FRT trees, the time complexity of FRT-embedding construction is bottlenecked by the construction of the \seq{}s.
Our efficient approach contains two subroutines:
\begin{itemize}
  \item An efficient (approximate) single-source shortest-path algorithm, introduced in Section~\ref{sec:sssp}.  The algorithm has three properties: linear complexity, distance preservation, and the ordering property  (full definitions given in Section~\ref{sec:sssp}).  All three properties are required for the correctness and efficiency of constructing FRT-embedding.  Our algorithm is a variant of Dijkstra's algorithm with the priority queue implemented by a new data structure called \emph{\struct}.
  \item An algorithm to integrate the shortest-path distances into the construction of FRT trees, discussed in Section~\ref{sec:seq}.  When the diameter of the graph is $n^{O(1)}$, we show that an FRT tree can be built directly using the approximate distances computed by shortest-path algorithm.  The challenge is when the graph diameter is large, and we proposed an algorithm that computes the approximate \seq{}s of a graph by concatenating the distances that only use the edges within a relative range of $n^{O(1)}$.  Then we show why the approximate \seq{}s still yield valid FRT trees.
\end{itemize}
With these new algorithmic subroutines, we show that the time complexity of computing FRT-embedding can be reduced to $O(m\log n)$ w.h.p.\ for an undirected positively weighted graph with arbitrary edge weight.

The second contribution of this paper is to show a new application of optimal probabilistic tree embeddings.
We show that FRT trees are intrinsically Ramsey partitions (definition given in Section~\ref{sec:frt-ramsey}) with asymptotically tight bound, and can achieve even better (constant) bounds on distance approximation.
Previous construction algorithms of optimal Ramsey partitions are based on
hierarchical CKR partitions, namely, on each level, the partition is
individually generated with an independent random radius and new
random priorities.  In this paper, we present a new proof to show that
the randomness in each level is actually unnecessary, so that only one
single random permutation is enough and the ratio of radii in
consecutive levels can be fixed as 2.
Our FRT-tree construction algorithm therefore can be directly applied to a number of
different distance oracles that are based on Ramsey partitions and accelerates the
construction of these distance oracles.

\hide{
Lastly, since the $O(n\log^2 n+m\log n)$ version of our algorithm to construct FRT trees is simple,
we actually implement it by applying it to a practical distance oracle by simply
taking multiple FRT trees of a graph.  We show that using 32 FRT trees on a number of real-world graphs,
the worst approximation for all queries that we tested is 3.3, the average approximation is 1.4, and the geometric mean is less than 1.07.
}
\hide{
We show that using $O(\log n)$
FRT trees, this distance oracle costs $O(m\log^2 n)$
preprocessing time, $O(n\log n)$ space, providing $O(\log
c)$-approximation w.h.p., and responding each query in $O(\log n)$ time,
where $c$ is the KR-expansion~\cite{karger2002} of the graph.
}
\hide{
Although the approximation bound is not particularly theoretically
interesting since $\log c$ can be as large as $\Theta(\log n)$ on
high-girth graphs or even social networks that obey power-law, we
implemented and tested it on a number of different types of graphs
(since this new distance oracle is simple and extremely easy to code).
Our experiments show that the distance approximation does well in
practice.  Using a small number (32) of FRT trees on graphs with both
large and small $c$, this distance oracle on average provides a
worst-case 3.3-approximation on one set of our testing cases.  Also,
the average approximation between pairs is 1.37 on the worst-case type
of graphs (with largest $c$) and the geometric mean of average
approximation is 1.07 across all our testing graphs.  We ran the
experiments multiple times with different seeds for the random number
generator, and the fluctuations of the numbers are small.
}

\section{Preliminaries and Notations}
\label{sec:prelim}

Let $G=(V,E)$ be a weighted graph with edge weights $w: E\rightarrow
\mathbb{R}_+$, and $d(u,v)$ denote the shortest-path distance in $G$
between nodes $u$ and $v$.
Throughout this paper, we assume that
$\min_{x\neq y}d(x,y)=1$. Let $\Delta={\max_{x,y}d(x,y)\over
  \min_{x\neq y}d(x,y)}=\max_{x,y}d(x,y)$, the diameter of the graph
$G$.

\hide{A \emph{hierarchically well-separated tree} (HST)~\cite{bartal96,Bar98} is a rooted tree for which edges in the same level have the same weight, and length of consecutive edges from the root to leaves decrease by factor larger than $1$.}

In this paper, we use the single source shortest paths problem (SSSP)
as a subroutine for a number of algorithms. Consider a weighted graph
with $n$ vertices and $m$ edges, Dijkstra's algorithm~\cite{Dijk}
solves the SSSP in $O(m+n\log n)$ time if the priority queue of
distances is maintained using a Fibonacci heap~\cite{Fib}.

A premetric $(X,d_X)$ defines on a set $X$ and provides a function $d:X\times X\to \mathbb{R}$ satisfying $d(x,x)=0$ and $d(x,y)\ge 0$ for $x,y\in X$.
A metric $(X,d_X)$ further requires $d(x,y)=0$ iff $x=y$, symmetry $d(x,y)=d(y,x)$, triangle inequality $d(x,y)\le d(x,z)+d(z,y)$ for $x,y,z\in X$.
The shortest-path distances on a graph is a metric and is called the graph metric and denoted as $d_G$.

We assume all intermediate results of our algorithm have word size $O(\log n)$ and basic algorithmic operations can be finished within a constant time.
Then within the range of $[1,n^k]$, the integer part of natural logarithm of an integer and floor function of an real number can be computed in constant time for any constant $k$.  This can be achieved using standard table-lookup techniques (similar approaches can be found in Thorup's algorithm~\cite{thorup1997}).
The time complexity of the algorithms are measured using the random-access machine (RAM) model.


A result holds \emph{with high probability} (\textbf{w.h.p.}) for an input of size $n$ if it holds with
probability at least $1-n^{-c}$ for any constant $c>0$, over all possible random choices made
by the algorithm.

Let $[n]=\{1,2,\cdots,n\}$ where $n$ is a positive integer.

Let $(X,d_X)$ be a metric space. For $Y\subseteq X$, define $(Y,d_X)$
as the metric $d_X$ restricted to pairs of points in $Y$, and
$\diam(Y,d_X)=\max\{d_X(x,y)\mid x,y\in Y\}$.  Define $B_X(x,r)=\{y\in
X\mid d_X(x,y)\leq r\}$, the closed ball centered at point $x$ and
containing all points in $X$ at a distance of at most $r$ from $x$. A
partition $\mathcal{P}$ of $X$ is a set of subsets of $X$ such that, for every
$x\in X$, there is one and only one unique element in $\mathcal{P}$,
denoted as $\mathcal{P}(x)$, that contains $x$.

The KR-expansion constant~\cite{karger2002} of a given metric space $(X,d_X)$ is defined as the smallest value of $c\ge 2$ such that $\left|B_X(x,2r)\right|\leq c\cdot\left|B_X(x,r)\right|$
for all $x\in X$ and $r > 0$.
The KR-dimension~\cite{karger2002} (or the expansion dimension) of $X$ is defined as $\dim_{\smb{KR}}(X)=\log c$.

We recall a useful fact about random permutations~\cite{Seidel93}:
\begin{lemma}
\label{lemma:logn}
Let $\pi:[n]\rightarrow[n]$ be a permutation selected uniformly at random on $[n]$. The set $\{i\mid i\in[n],\pi(i)=\min\{\pi(j)\mid j=1,\cdots, i\}\}$ contains $O(\log n)$ elements both in expectation and with high probability.
\end{lemma}

\hide{
\myparagraph{Probabilistic tree embeddings}
Given a metric space $(X,d_X)$, the goal of probabilistic tree
embeddings is to find a distribution of dominating trees~\cite{Bartal98}
with a minimum distance distortion.  Let $T$ be a tree with weighted
edges and with leaves corresponding to the elements in the metric
space.  Denote $d_T(x,y)$ as the tree metric, which represents the
distance of the path connecting $x$ and $y$ in the tree.  A
distribution $\mathcal{D}$ of trees should follow the dominating
property and minimize the expected stretch $\psi$.  \hide{ so that for
  each tree $T\in \mathcal{D}$:
\begin{enumerate}
  \item $d_X(x,y)\leq d_T(x,y)$ for every $x,y\in X$;
  \item $\E[d_T(x,y)]\leq \psi\cdot d_X(x,y)$ for every $x,y\in X$,
\end{enumerate}}
The FRT algorithm~\cite{FRT} is an optimal algorithm with $\psi=O(\log n)$, and this bound is asymptotically tight for certain classes of finite metric spaces.
}

\paragraph{Ramsey partitions.}
Let $(X,d_X)$ be a metric space.  A hierarchical partition tree of $X$ is a
sequence of partitions $\{\mathcal{P}_k\}_{k=0}^\infty$ of $X$ such
that $\mathcal{P}_0=\{X\}$, the diameter of the partitions in each
level decreases by a constant $c>1$, and each level
$\mathcal{P}_{k+1}$ is a refinement of the previous level
$\mathcal{P}_k$.  A Ramsey partition~\cite{Ramsey} is a distribution
of hierarchical partition trees such that each vertex has a lower-bounded probability
of being sufficiently far from the partition boundaries in all
partitions $k$, and this gap is called the \emph{padded range} of a
vertex.  More formally:
\begin{definition}
\label{def:Ramsey}
An $(\alpha,\gamma)$-Ramsey partition of a metric space $(X,d_X)$ is a probability distribution over hierarchical partition trees $\mathcal{P}$ of $X$ such that for every $x\in X$:
$$\Pr\left[\forall k\in \mathbb{N},B_X\left(x,\alpha\cdot c^{-k}\Delta\right)\subseteq \mathcal{P}_k(x)\right]\geq|X|^{-\gamma}.$$
\end{definition}
An asymptotically tight construction of Ramsey partition where $\alpha=\Omega(\gamma)$ is provided by Mendel and Naor~\cite{Ramsey} using the Calinescu-Karloff-Rabani partition~\cite{CKR} for each level.

\paragraph{Approximate distance oracles.}
Given a finite metric space $(X,d_X)$, we want to support efficient
approximate pairwise shortest-distance queries.  Data structures to
support this query are called approximate distance oracles\hide{, a
  ``compact'' version of distances}.  A $(P,S,Q,D)$-distance oracle on
a finite metric space $(X,d_X)$ is a data structure that takes
expected time $P$ to preprocess from the given metric space, uses $S$
storage space, and answers distance query between points $x$ and $y$
in $X$ in time $Q$ satisfying $d_X(x, y)\leq d_O(x, y)\leq D\cdot d_X(x,
y)$, where $d_O(x, y)$ is the pairwise distance provided by the distance
oracle.

The concept of approximate distance oracles was first studied by
Thorup and Zwick~\cite{TZ}.  Their distance oracles are based on
graph spanners, and given a graph, a $(O(kmn^{1/k}), O(kn^{1+1/k}),
O(k), 2k-1)$-oracle can be created.  This was followed by many
improved results, including algorithms focused on distance oracles
that provide small stretches ($<3$)~\cite{baswana2006,agarwal2013},
and oracles that can report paths in addition to
distances~\cite{elkin2015}.  A recent result of
Chechik~\cite{chechik15} provides almost optimal space, stretch and
query time, but the construction time is polynomially slower than the
results in this paper, and some of the other work, when $k>2$.

For distance oracles that can be both constructed and queried
efficiently on a graph, a series of algorithms, including
Mendel-Naor's~\cite{Ramsey}, Wulff-Nilsen's~\cite{WN} and
Chechik's~\cite{SC}, all use the Ramsey partitions of a graph as an algorithmic building block
of the algorithm.


\section{An Approximate SSSP Algorithm}
\label{sec:sssp}

In this section we introduce a variant of Dijkstra's algorithm.
This is an efficient algorithm for single-source shortest paths (SSSP) with linear time complexity $O(m)$.
The computed distances are $\alpha$-\disPreserve{}:
\begin{definition}[$\alpha$-\disPreserve{}]
For a weighted graph $G=(V,E)$, the single-source distances $d(v)$ for $v\in V$ from the source node $s$ is $\alpha$-\disPreserve{}, if there exists a constant $0\le\alpha\le 1$ such that
$\alpha\, d_G(s,u)\le d(u)\le  d_G(s,u)$, and $d(v)-d(u)\le d_G(u,v)$, for every $u,v\in V$.
\end{definition}
$\alpha$-\disPreserve{} can be viewed as the triangle inequality on single-source distances (i.e.\ $d(u)+d_G(u,v)\ge d(v)$ for $u,v\in V$), and is required in many applications related to distances.
For example, in Corollary~\ref{cor:gabow} we show that using Gabow's scaling algorithm~\cite{gabow1985scaling} we can  compute a ($1-\epsilon$)-approximate SSSP using $O(m\log \epsilon^{-1})$ time.
Also in many metric problems including the contruction of optimal tree embeddings, distance preservation is necessary in the proof of the expected stretch, and such an example is Lemma~\ref{lem:approx-dis} in Section~\ref{sec:expected-stretch}.

The preliminary version we discussed in Section~\ref{sec:algo-prelim} limits edge weights in $[1,n^k]$ for a constant $k$, but with some further analysis in 
the full version of this paper
we can extend the range to $[1, n^{O(m)}]$.
This new algorithm also has two properties that are needed in the construction of FRT trees, while no previous algorithms achieve them all:

\begin{enumerate}
  \item ($\alpha$-\disPreserve{}) The computed distances from the source $d(\cdot)$ is $\alpha$-\disPreserve{}.
  \item (Ordering property and linear complexity) The vertices
are visited in order of distance $d(\cdot)$, and the time to compute the first $k$ distances is bounded by $O(m')$ where $m'$ is the sum of degrees from these $k$ vertices.
\end{enumerate}

The algorithm also works on directed graphs, although this is not used in the FRT construction.

Approximate SSSP algorithms are well-studied~\cite{thorup1997,klein1997,cohen2000,pettie2005,miller2015}.
In the sequential setting, Thorup's algorithm~\cite{thorup1997} compute single-source distances on undirected graphs with integer weights using $O(n+m)$ time.
Nevertheless, Thorup's algorithm does not obey the ordering property since it uses a hierarchical bucketing structure and does not visit vertices in an order of increasing distances, and yet we are unaware of a simple argument to fix this.
Other algorithms are either not work-efficient (i.e.\ super-linear complexity) in sequential setting, and / or violating distance preservation.

\begin{theorem}\label{thm:sp}
For a weighted directed graph $G=(V,E)$ with edge weights between $1$ and $n^{O(1)}$,
a $(1/4)$-\disPreserve{} single-source shortest-path distances $d(\cdot)$ can be computed, such that the distance to the $k$-nearest vertices $v_1$ to $v_k$ by $d(\cdot) $requires $O(\sum_{i=1}^{k}{\mb{degree}(v_i)})$ time.
\end{theorem}

The algorithm also has the two following properties.
We discuss how to (1) extend the range of edge weights to $n^{O(m)}$, and the cost to compute the $k$-nearest vertices is $O(\log_n{d(v_k)}+\sum_{i=1}^{k}{\mb{degree}(v_i)})$ where $v_1$ to $v_k$ are the $k$ nearest vertices (Section~\ref{sec:algo-ext}); and (2) compute $(1+\epsilon)$-distance-preserving shortest-paths for an arbitrary $\epsilon>0$:

\begin{corollary}\label{cor:gabow}
For a graph $G = (V,E)$ and any source $s \in V$, $(1-\epsilon)$-distance-preserving approximate distances $d(v)$ for $v \in V$ from $s$, can be computed by repeatedly using the result of Theorem~\ref{thm:sp} $O(\log \epsilon^{-1})$ rounds, which leads to $O(m log \epsilon^{-1})$ when edge weights are within $n^{O(1)}$.

$(1-\epsilon)$-distance-preserving shortest-paths for all vertices can be computed by repeatedly using Theorem~\ref{thm:sp} $O(\log \epsilon^{-1})$ rounds.
\end{corollary}
\begin{proof}
To get $(1-\epsilon)$-distance-preserving shortest-paths, we can use the algorithm in Theorem~\ref{thm:sp} $O(\log \epsilon^{-1})$ rounds repeatedly, and the output shortest-paths in $i$-th round are denoted as $d_i(\cdot)$.  
The first round is run on the input graph.
In the $i$-th round for $i>1$, each edge $e$ from $u$ to $v$ in the graph is reweighted as $w_e+\sum_{j=1}^i(d_j(u)-d_j(v))$.
Since we know that the output is \disPreserve{} for the input graph in each round, inductively we can check that all reweighted edges are positive.  
This also indicates that $\sum_{j=1}^i d_j(u)\le d(s,u)$.

We now show that, from source node $s$, $d(s,u)-\sum_{j=1}^i d_j(u)\le (1-\alpha)^i \cdot d(s,u)$.  This is because, by running the shortest-path algorithm that is $\alpha$-\disPreserve{}, $d_i(u)\ge \alpha \cdot(d(s,u)-\sum_{j=1}^{i-1} d_j(u))$ based on the way the graph is reweighted.
Therefore, the summation of all $d_i(u)$ is at least $(1-\epsilon)\,d(s,u)$ after $O(\log \epsilon^{-1})$ rounds.
\end{proof}

\subsection{Algorithm Details}\label{sec:algo-prelim}

The key data structure in this algorithm is a \struct{} shown in Figure~\ref{fig:bucket} that implements the priority queue in Dijkstra's algorithm.
With the \struct{}, each \mf{Decrease-Key} or \mf{Extract-Min} operation takes constant time.
Given the edge range in $[1,n^k]$, this structure has $l=\lceil (1+k)\log_2{n}\rceil$ levels, each level containing a number of buckets corresponding to the distances to the source node.
In the lowest level (level 1) the difference between two adjacent buckets is $2$.

At anytime only one of the buckets in each level can be non-empty: there are in total $l$ active buckets to hold vertices, one in each level.
The active bucket in each level is the left-most bucket whose distance is larger than that of the current vertex being visited in our algorithm.
We call these active buckets the \textbf{\emph{frontier}} of the current distance, and they can be computed by the \textbf{\emph{path string}}, which is a 0/1 bit string corresponding to the path from the current location to higher levels (until the root), and 0 or 1 is decided by whether the node is the left or the right child of its parent.
For clarity, we call the buckets on the frontier \textbf{\emph{frontier buckets}}, and the ancestors of the current bucket \textbf{\emph{ancestor buckets}} (can be traced using the path string).
For example, as the figure shows, if the current distance is 4, then the available buckets in the first several levels are the buckets corresponding to the distances $6, 5, 11, 7, $ and so on.
The ancestor bucket and the frontier bucket in the same level may or may not be the same, depending on whether the current bucket is the left or right subtree of this bucket.
For example, the path string for the current bucket with label 4 is $0100$ and so on, and ancestor buckets correspond to $4, 5, 3, 7$ and so on.
It is easy to see that given the current distance, the path string, the ancestor buckets, and the frontier buckets can be computed in $O(l)$ time---constant time per level.

Note that since only one bucket in each level is non-empty, the whole structure need not to be build explicitly: we store one linked list for each level to represent the only active bucket in the memory ($l$ lists in total), and use the current distance and path string to retrieve the location of the current bucket in the structure.

\begin{figure}
\centering
  \includegraphics[width=.8\columnwidth]{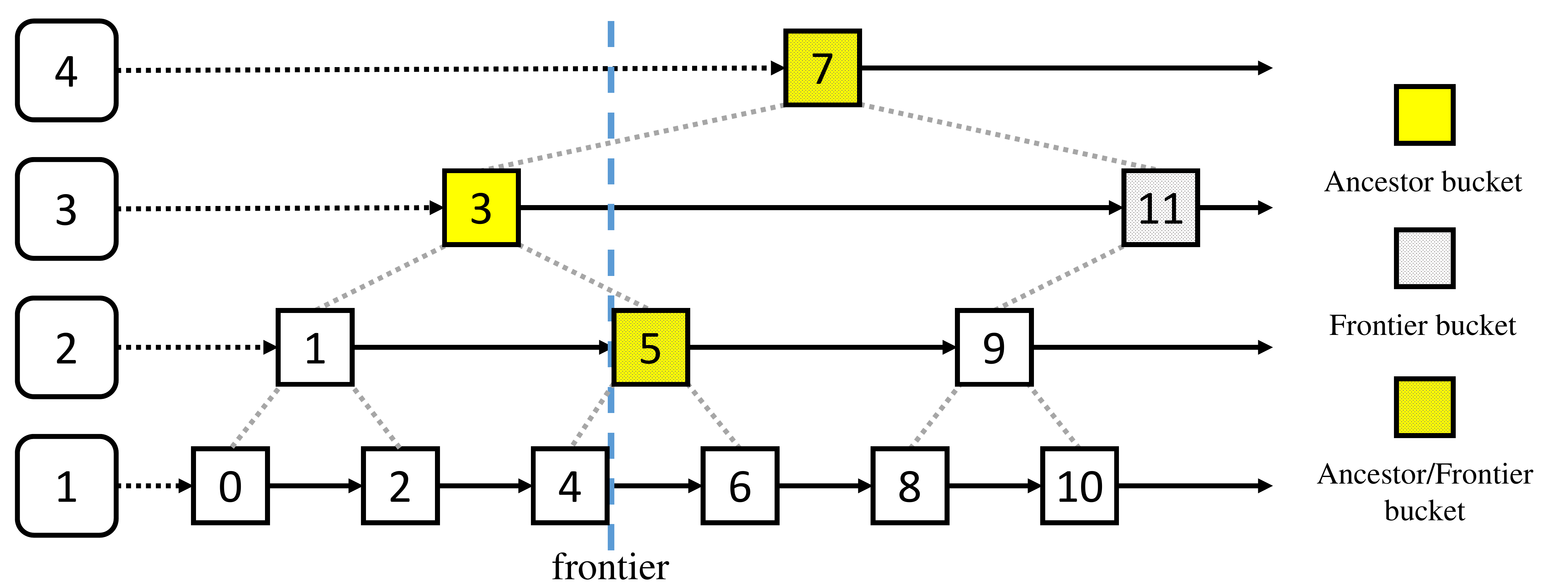}
  \caption{An illustration of a \struct{} with the lowest 4 levels, and the current visiting bucket has distance 4.  Notice that our algorithm does not insert vertices to the same level as the current bucket (i.e.\ bucket 6).}\label{fig:bucket}
  \vspace{-1em}
\end{figure}

With the \struct{} acting as the priority queue, we can run standard Dijkstra's algorithm.
The only difference is that, to achieve linear cost for an SSSP query, the operations of \mf{Decrease-Key} and \mf{Extract-Min} need to be redefined on the \struct{}.

Once the relaxation of an edge succeeds, a \mf{Decrease-Key} operation for the corresponding vertex will be applied.
In the \struct{} it is implemented by a \mf{Delete} (if the vertex is added before) followed by an \mf{Insert} on two frontier buckets respectively.
The deletion is trivial with a constant cost, since we can maintain the pointer from each vertex to its current location in the \bucket{}.
We mainly discuss how to insert a new tentative distance into the \bucket{}.
When vertex $u$ successfully relaxes vertex $v$ with an edge $e$, we first round down the edge weight $w_e$ by computing $r=\lfloor\log_2{(w_e+1)}\rfloor$. 
Then we find the appropriate frontier bucket $B$ that the difference of the distances $w_e'$ between this bucket $B$ and the current bucket is the closest to (but no more than) $w_r=2^r-1$, and insert the relaxed vertex into this bucket.
The constant approximation for this insertion operation holds due to the following lemma:
\begin{lemma}\label{lemma:woverfour}
For an edge with length $w_e$, the approximated length $w_e'$, which is the distance between the inserted bucket $B$ and the current bucket, satisfies the following inequality: $w_e/4\le w_e'\le w_e$.
\end{lemma}
\begin{proof}
After the rounding, $w_r=2^r-1=2^{\lfloor\log_2{(w_e+1)}\rfloor}-1$ falls into the range of $[w_e/2,w_e]$.
We now show that there always exists such a bucket $B$ on the frontier that the approximated length $w_e'$ is in $[w_r/2,w_r]$.

We use Algorithm~\ref{algo:bucket} to select the appropriate bucket for a certain edge, given the current bucket level and the path string.
The first case is when $b$, the current level, is larger than $r$.
In this case all the frontier buckets on the bottom $r$ levels form a left spine of the corresponding subtree rooted by the right child of the current bucket, so picking bucket in the $r$-th level leads to $w_e'=2^{r-1}$, and therefore $w_e/4<w_e'\le w_e$ holds.
The second case is when $b\le r$, and the selected bucket is decided based on the structure on the ancestor buckets from the $(r+1)$-th level to $(r-1)$-th level, which is one of the three following cases.
\begin{itemize}
  \item The simplest case ($b<r$, line~\ref{line:case1}) is when the ancestor bucket in the $(r-1)$-th level is the right child of the bucket in the $r$-th level.
      In this case when we pick the bucket in level $r$ since the distance between two consecutive buckets in level $r$ is $2^r$, and the distance from the current bucket to the ancestor bucket in $r$-th level is at most $\sum_{i=1}^{r-1}{2^{i-1}}<2^{r-1}$.
      The distance thus between the current bucket and the frontier bucket in level $r$ is $w_e'>2^r-2^{r-1}=2^{r-1}>w_e/4$.
  \item The second case is when either $b=r$ and the current bucket is the left child (line~\ref{line:case3}), or $b<r$ and the ancestor bucket in level $r-1$ is on the left spine of the subtree rooted at the ancestor bucket in level $r+1$ (line~\ref{line:case2}).
      Similar to the first case, picking the frontier bucket in the $(r+1)$-th level (which is also an ancestor bucket) skips the right subtree of the bucket in $r$-th level, which contains $2^{r-1}-1\ge w_e/4$ nodes.
  \item The last case is the same as the second case expect that the level-$r$ ancestor bucket is the right child of level-$(r+1)$ ancestor bucket.
      In this case we will pick the frontier bucket that has distance $2^{r-1}$ to the ancestor bucket in level $r$, which is the parent of the lowest ancestor bucket that is a left child and above level $r$.  In this case the approximated edge distance is between $2^{r-1}$ and $2^r-1$.
\end{itemize}
Combining all these cases proves the lemma.
\end{proof}

\begin{algorithm}[!tp]
\caption{Finding the appropriate bucket}
\label{algo:bucket}
\KwIn{Current bucket level $b$, rounded edge length $2^r-1$ and path string.}
\KwOut{The bucket in the frontier (the level is returned).}
    \vspace{0.5em}
    Let $r'$ be the lowest ancestor bucket above level $r$ that is a left child\\
    \DontPrintSemicolon
    \If{$b> r$} {\Return {$r$}}
    \ElseIf {$b=r$} { \lIf {current bucket is left child} {\label{line:case3}\Return $r+1$} \lElse {\label{line:case4}\Return $r'+1$} }
    \Else {
            \Switch {the branches from $(r+1)$-th level to $(r-1)$-th level in the path string} {
                \Case {left-then-right or right-then-right\label{line:case1}}
                    {\Return $r$}
                \Case {left-then-left\label{line:case2}}
                    {\Return $r+1$}
                \Case {right-then-left}
                    {\Return $r'+1$}
        }
    }
\end{algorithm}

We now explain ehe \mf{Extract-Min} operation on the \bucket{}.
We will visit vertex in the current buckets one by one, so each \mf{Extract-Min} has a constant cost.
Once the traversal is finished, we need to find the next closest non-empty frontier.
\begin{lemma}\label{lemma:cost-lb}
\mf{Extract-Min} and \mf{Decrease-Key} on the \bucket{} require $O(1)$ time.
\end{lemma}
\begin{proof}
We have shown that the modification on the linked list for each operation requires $O(1)$ time.
A na\"ive implementation to find the bucket in \mf{Decrease-Key} and \mf{Extract-Min} takes $O(l)=O(\log n)$ time, by checking all possible frontier buckets.
We can accelerate this look-up using the standard table-lookup technique.
The available combinations of the input of \mf{Decrease-Key} are $n^{k+1}$ (total available current distance) by $l=O(k\log n)$ (total available edge distance after rounding), and the input combinations of \mf{Extract-Min} are two $\lceil \log_2{n^{k+1}}\rceil$ bit strings corresponding to the path to the root and the emptiness of the buckets on the frontier.
We therefor partition the \bucket{} into several parts, each containing $\lfloor (1-\epsilon')(\log_2{n})/2\rfloor$ consecutive levels (for any $0<\epsilon'<1$).
We now precompute the answer for all possible combinations of path strings and edge lengths, and (1) the sizes of look-up tables for both operations to be $O((2^{\lfloor (1-\epsilon')(\log_2{n})/2\rfloor})^2)=o(n)$, (2) the cost for brute-force preprocessing to be $O((2^{\lfloor (1-\epsilon')(\log_2{n})/2\rfloor})^2\log n)=o(n)$, and (3) the time of either operation of \mf{Decrease-Key} and \mf{Extract-Min} to be $O(k)$, since each operation requires to look up at most $l/\lfloor (1-\epsilon')(\log_2{n})/2\rfloor=O(k)$ tables.
Since $k$ is a constant, each of the two operations as well takes constant time.
The update of path string can be computed similarly using this table-lookup approach.
As a result, with $o(n)$ preprocessing time, finding the associated bucket for \mf{Decrease-Key} or \mf{Extract-Min} operation uses $O(1)$ time.
\end{proof}

We now show the three properties of the new algorithm: linear complexity, \disPreserve{}, and the ordering property.

\begin{proof}[Proof of Theorem~\ref{thm:sp}]
Here we show the algorithm satisfies the properties in Theorem~\ref{thm:sp}.
Lemma~\ref{lemma:cost-lb} proves the linear cost of the algorithm.
Lemma~\ref{lemma:woverfour} shows that the final distances is $\alpha$-\disPreserve{}.
Lastly, since this algorithm is actually a variant of Dijkstra's algorithm with the priority implemented by the \struct{}, the ordering property is met, although here the $k$-nearest vertices are based on the approximate distances instead of real distances.
\end{proof}

\subsection{Extension to greater range for edge weights}\label{sec:algo-ext}

We now discuss how to extend the range of edge weight to $[1,n^{O(m)}]$.
For the larger range of edge weights, we extend the \bucket{} to contains $O(m\log n)$ levels in total.
Note that at anytime only $O(m)$ of them will be non-empty, so the overall storage space is $O(m)$.
In this case there exists an extra cost of $O(\log_n{d(v_k)})$ when computing the shortest-paths to the $k$-nearest vertices.
This can be more than $O(\sum_{i=1}^{k}{\mb{degree}(v_i)})$, so for the FRT-tree construction in the next section, we only call SSSP queries with edge weight in a relative range of $n^{O(1)}$.

The challenge of this extension here is that, the cost of table-lookup for \mf{Decrease-Key} and \mf{Extract-Min} now may exceed a constant cost.
However, the key observation is that, once the current distance from the source to the visiting vertex is $d$, all edges with weight less than $d/n$ can be considered to be 0.
This is because the longest path between any pair of nodes contains at most $n-1$ edges, which adding up to $d$.
Ignoring the edge weight less than $d/n$ will at most include an extra factor of 2 for the overall distance approximation (or $1+\epsilon'$ for an arbitrarily small constant $\epsilon'>0$ if the cutoff is $d\epsilon'/n$).  Thus the new vertices relaxed by these zero-weight edges are added back to the current bucket, instead of the buckets in the lower levels.  For \mf{Extract-Min}, once we have visited the bucket in the $b$-th level, the bucket in level $b'<b-\lfloor\log_2 n\rfloor$ will never be visited, so the amortized cost for each query is constant.
For \mf{Decrease-Key}, similarly if the current distance is $d$, the buckets below level $\lfloor\log_2 {d/n}\rfloor$ will never have vertices to be added, and buckets above level $\lceil \log_2 {d+1}\rceil$ are always the left children of their parents.  The number of the levels that actually require table lookup to be preprocessed is $O(\log n)$, so by applying the technique introduced in Section~\ref{sec:algo-prelim}, the cost for one \mf{Decrease-Key} is also linear.




\section{The Dominance Sequence}
\label{sec:seq}

In this section we review and introduce the notion of \seq{}s for each point of a
metric space and describe the algorithm for constructing them on a graph.
The basic idea of \seq{}s was previously introduced in~\cite{cohen1997} and~\cite{Kanat}.
Here we name the structure as the \seq{} since the ``dominance'' property introduced below is crucial and related to FRT construction.
In the next section we show how they can easily be converted into an FRT tree.

\subsection{Definition}\label{sec:seq-defn}



\begin{definition}[Dominance]
Given a premetric $(X,d_X)$ and a permutation $\pi$, for two points $x,y\in X$, $x$ \emph{dominates} $y$ if and only if
$$\pi(x)=\min\{\pi(w)\mid w\in X, d_X(w,y)\leq d_X(x,y)\}.$$
\end{definition}

Namely, $x$ dominates $y$ iff $x$'s priority is greater (position in the permutation is earlier)
than any point that is closer to $y$.

The \seq{} for a point $x\in X$, is the sequence of all points that
dominate $x$ sorted by distance.  More formally:
\begin{definition}[Dominance Sequence\footnote{Also called as ``least-element list'' in~\cite{cohen1997}.  We rename it since in later sections we also consider many other variants of it based on the dominance property.}]
For each $x\in X$ in a premetric $(X,d_X)$, the \emph{\seq}
of a point $x$ with respect to a permutation $\pi:X\rightarrow [n]$
(denoted as $\chi_{\pi}^{(x)}$), is the sequence
$\langle p_i\rangle_{i=1}^k$ such that
$1=\pi(p_1)<\pi(p_2)<\cdots<\pi(p_k)=\pi(x)$, and $p_i$ is in
$\chi_{\pi}^{(x)}$ iff $p_i$ dominates $x$.
\end{definition}

We use $\chi_{\pi}$ to refer to all \seq{}s for a premetric under permutation $\pi$.  It is not hard to bound the size of the \seq:
\begin{lemma}[\cite{cohen2000}]
\label{lemma:sigma}
Given a premetric $(X,d_X)$ and a random permutation $\pi$, for each vertex $x\in X$, with w.h.p.
$$\left|\chi_{\pi}^{(x)}\right|= O(\log n)$$
and hence overall, with w.h.p.
$$\left|\chi_{\pi}\right|=\sum_{x\in X}\left|\chi_{\pi}^{(x)}\right|= O(n\log n)$$
\end{lemma}
Since the proof is fairly straight-forward, for completeness we also provide it in the full version of this paper.

Now consider a graph metric $(V,d_G)$ defined by an undirected positively weighted graph $G=(V,E)$ with $n$ vertices and $m$ edges, and $d_G(u,v)$ is the shortest distance between $u$ and $v$ on $G$.
The \seq{}s of this graph metric can be constructed using $O(m\log n+n\log^2 n)$ time w.h.p.\,\cite{cohen1997}.
This algorithm is based on Dijkstra's algorithm.

\subsection{Efficient FRT tree construction based on the \seq{}s}
\label{sec:tree-constr}

We now consider the construction of FRT trees based on a pre-computed
\seq{}s of a given metric space $(X,d_X)$.
\hide{
Here we assume a graph metric based on an undirected positively
weighted graph $G=(V,E)$ containing $n=|V|$ vertices and $m=|E|$
edges, where the metric $d_G(u,v)$ is the shortest distance between
vertices $u$ and $v$. }
We assume the weights are normalized so
that $1\leq d_X(x,y)\leq\Delta=2^{\delta}$ for all $x\ne y$, where
$\delta$ is a positive integer.

The FRT algorithm~\cite{FRT} generates a top-down recursive
low-diameter decomposition (LDD) of the metric, which preserves the
distances up to $O(\log n)$ in expectation.
It first chooses a
random $\beta$ between 1 and 2, and generates $1+\log_2 \Delta$ levels of
partitions of the graph with radii
$\{\beta\Delta,\beta\Delta/2,\beta\Delta/4,\cdots\}$.  This procedure
produces a laminar family of clusters, which are connected based on
set-inclusion to generate the FRT tree.  The weight of each tree edge on level $i$
is $\beta\Delta/2^{i}$.

Instead of computing these partitions directly, we adopt the idea of a
point-centric view proposed in~\cite{Kanat}.  We use the intermediate
data structure ``\seq{}s'' as introduced in Section~\ref{sec:seq-defn} to
store the useful information for each point.  Then, an FRT tree can be
retrieved from this sequence with very low cost:

\begin{lemma}\label{lem:convert}
Given $\beta$ and the \seq{}s $\chi_{\pi}$ of a metric space with associated
distances to all elements, an FRT tree can be constructed using $O(n \log n)$ time w.h.p.
\end{lemma}

The difficulty in this process is that, since the FRT tree
has $O(\log \Delta)$ levels and $\Delta$ can be large
(i.e.\ $\Delta>2^{O(n)}$), an explicit representation of the FRT tree can be very costly.
Instead we generate the compressed version with nodes
of degree two removed and their incident edge weights summed into a new edge.
The algorithm is outlined in Algorithm~\ref{algo:FRT}.

\begin{algorithm}[!tp]
\caption{Efficient FRT tree construction}
\label{algo:FRT}
    Pick a uniformly random permutation $\pi:V\rightarrow [n]$.\\
    Compute the \seq{}s $\chi_{\pi}$.\\
    Pick $\beta\in[1,2]$ with the probability density function $f_B(x)=1/(x\ln 2)$.\\
    Convert the \seq~$\chi_{\pi}$ to the \cps~$\bar{\sigma}_{\pi,\beta}$.\\
    Generate the FRT tree based on $\bar{\sigma}_{\pi,\beta}$.\\
\end{algorithm}

\begin{figure}[t]
\vspace{-1em}
\centering
  \includegraphics[width=\columnwidth]{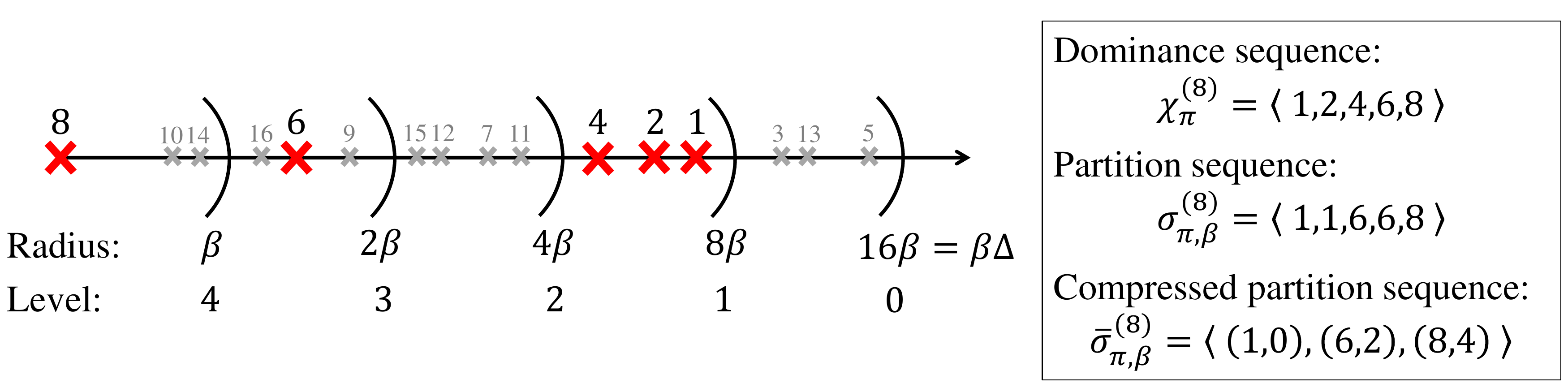}
\caption{An illustration for dominance sequence, partition sequence and compressed partition sequence for vertex 8. Here we assume that the label of each vertex corresponds to its priority. The left part shows the distances of all vertices to vertex 8 in log-scale, and the red vertices dominate vertex 8.}\label{graphs}
\label{fig:sequences}
\vspace{-1em}
\end{figure}

\begin{proof}
We use the definition of \pars{} and \cps{} from~\cite{Kanat}.  Given a
permutation $\pi$ and a parameter $\beta$, the \emph{\pars{}} of a
point $x\in X$, denoted by $\sigma_{\pi,\beta}^{(x)}$, is the
sequence $\displaystyle \sigma_{\pi,\beta}^{(x)}(i) = \min \{
\pi(y)\mid y \in X, d(x,y) \leq \beta \cdot 2^{\delta-i}\}$ for $i =
0, \ldots, \delta$, i.e.\ point $y$ has the highest priority among
vertices up to level $i$.  We note that a trie (radix tree) built on the \pars{} is
the FRT tree, but as mentioned we cannot build this explicitly. The
\cps, denoted as $\bar{\sigma}_{\pi,\beta}^{(x)}$, replaces
consecutive equal points in the \pars{} $\sigma_{\pi,\beta}^{(x)}$ by the
pair $(p_i,l_i)$ where $p_i$ is the vertex and $l_i$ is the highest
level $p_i$ dominates $x$ in the FRT tree.  Figure~\ref{fig:sequences} gives an example of
a \pars{}, a \cps{}, and their relationship to the \seq{}.

To convert the \seq{}s $\chi_{\pi}$ to the \cps{s}
$\bar{\sigma}_{\pi,\beta}$ note that for each point $x$ the points
in $\bar{\sigma}_{\pi,\beta}^{(x)}$ are a subsequence of
$\chi_{\pi}^{(x)}$.  Therefore, for $\bar{\sigma}_{\pi,\beta}^{(x)}$,
we only keep the highest priority vertex in each level from
$\chi_{\pi}^{(x)}$ and tag it with the appropriate level.  Since there
are only $O(\log n)$ vertices in $\chi_{\pi}^{(x)}$ w.h.p., the time
to generate $\bar{\sigma}_{\pi,\beta}^{(x)}$ is $O(\log n)$ w.h.p., and hence
the overall construction time is $O(n\log n)$ w.h.p.

The compressed FRT tree can be easily generated from the \cps{s}
$\bar{\sigma}_{\pi,\beta}$.  Blelloch et. al.~\cite{Kanat} describe a
parallel algorithm that runs in $O(n^2)$ time (sufficient for their
purposes) and polylogarithmic depth.  Here we describe a similar version
to generate the FRT tree sequentially in $O(n \log n)$ time
w.h.p.  The idea is to maintain the FRT as a patricia
trie~\cite{Morrison68} (compressed trie) and insert the \cps{s} one at
a time.  Each insertion just needs to follow the path down the tree
until it diverges, and then either split an edge and create a new node,
or create a new child for an existing node.
Note that a hash table is required to trace the tree nodes since the trie has a non-constant alphabet.
Each insertion takes time
at most the sum of the depth of the tree and the length of the
sequence, giving the stated bounds.
\end{proof}

We note that for the same permutation $\pi$
and radius parameter $\beta$, it generates exactly the same tree as the
original algorithm in~\cite{FRT}.

\subsection{Expected Stretch Bound}\label{sec:expected-stretch}

In Section~\ref{sec:tree-constr} we discussed the algorithm to convert the \seq{s} to a FRT tree.
When the \seq{s} is generated from a graph metric $(G,d_G)$, the expected stretch is $O(\log n)$, which is optimal, and the proof is given in many previous papers~\cite{FRT,Kanat}.
Here we show that any distance function $\hat{d}_G$ in Lemma~\ref{lem:approx-dis} is sufficient to preserve this expected stretch.
As a result, we can use the approximate shortest-paths computed in Section~\ref{sec:sssp} to generate the \seq{s} and further convert to optimal tree embeddings.

\begin{lemma}\label{lem:approx-dis}
Given a graph metric $(G,d_G)$ and a distance function $\hat{d}_G(u,v)$ such that for $u,v,w\in V$, $\lvert \hat{d}_G(u,v)-\hat{d}_G(u,w)\rvert\le 1/\alpha\cdot d_G(v,w)$ and $d_G(u,v)\le\hat{d}_G(u,v)\le 1/\alpha\cdot d_G(u,v)$ for some constant $0<\alpha\le 1$, then the \seq{s} based on $(G,\hat{d}_G)$ can still yield optimal tree embeddings.
\end{lemma}

\begin{proof}[Proof outline]
Since the overestimate distances hold the dominating property of the tree embeddings, we show the expected stretch is also not affected.
We now show the expected stretch is also held.

Recall the proof of the expected stretch by Blelloch et al.\ in~\cite{Kanat} (Lemma 3.4).
By replacing $d_G$ by $\hat{d}_G$, the rest of the proof remains unchanged except for Claim 3.5, which upper bounds the expected cost of a common ancestor $w$ of $u,v\in V$ in $u$ and $v$'s \seq{s}.
The original claim indicates that the probability that $u$ and $v$ diverges in a certain level centered at vertex $w$ is $O(\lvert d_G(w,u)-d_G(w,v)\rvert/d_G(u,w))=O(d_G(u,v)/d_G(u,w))$ and the penalty is $O(d_G(u,w))$, and therefore the contribution of the expected stretch caused by $w$ is the product of the two, which is $O(d_G(u,v))$ (since there are at most $O(\log n)$ of such $w$ (Lemma~\ref{lemma:sigma}), the expected stretch is thus $O(\log n)$).
With the distance function $\hat{d}_G$ and $\alpha$ as a constant, the probability now becomes $O(\lvert \hat{d}_G(w,u)-\hat{d}_G(w,v)\rvert/\hat{d}_G(u,w))=O(d_G(u,v)/d_G(u,w))$, and the penalty is $O(\hat{d}_G(u,w))=O(d_G(u,w))$.
As a result, the expected stretch asymptotically remains unchanged.
\end{proof}

\subsection{Efficient construction of approximate \seq{}s}\label{sec:eff-constr-app-seq}

Assume that $\hat d_G(u,v)$ is computed as $d_u(v)$ by the shortest-path algorithm in Section~\ref{sec:sssp} from the source node $u$.
Notice that $d_u(v)$ does not necessarily to be the same as $d_v(u)$, so $(G,\hat d_G(u,v))$ is not a metric space.
Since the computed distances are distance preserving, it is easy to check that Lemma~\ref{lem:approx-dis} is satisfied, which indicate that we can generate optimal tree embeddings based on the distances.
This leads to the main theorem of this section.

\begin{theorem} \emph{\textbf{(Efficient optimal tree embeddings)}}
There is a randomized algorithm that takes an undirected positively weighted graph $G=(V,E)$ containing $n=|V|$ vertices and $m=|E|$ edges, and produces an tree embedding such that for all $u,v\in V$, $d_G(u,v)\leq d_T(u,v)$ and $\E[d_T(u,v)]\leq O(\log n)\cdot d_G(u,v)$. The algorithm w.h.p.\,runs in $O(m\log n)$ time.
\end{theorem}

The algorithm computes approximate \seq{s} $\hat\chi_{\pi}$ by the approximate SSSP algorithm introduced in Section~\ref{sec:sssp}.
Then we apply Lemma~\ref{lem:convert} to convert $\hat\chi_{\pi}$ to an tree embedding.
Notice that this tree embedding is an FRT-embedding based on $\hat d_G$.
We still call this an FRT-embedding since the overall framework to generate hierarchical tree structure is similar to that in the original paper~\cite{FRT}.

The advantage of our new SSSP algorithm is that the \mf{Decrease-Key} and \mf{Extract-Min} operation only takes a constant time when the relative edge range (maximum divided by minimum) is no more than $n^{O(1)}$.
Here we adopt the similar idea from~\cite{klein1997} to solve the subproblems on specific edge ranges, and concatenate the results into the final output.
We hence require a pre-process to restrict the edge range.
The pseudocode is provided in Algorithm~\ref{algo:linear}, and the details are explained as follows.

\myparagraph{Pre-processing \normalfont{(line 1--\ref{line:endofprepro} in Algorithm \ref{algo:linear})}}  The goal is to use $O(m\log n)$ time to generate a list of subproblems: the $i$-th subproblem has edge range in the interval $[n^{i-1}, n^{i+1}]$ and we compute the elements in the \seq{}s with values falling into the range from $n^i$ to $n^{i+1}$.
All the edge weights less than the minimum value of this range are treated as 0.
Namely, the vertices form some components in one subproblem and the vertex distances within each component is 0.
We call this component in a specific subproblem the \textbf{vertex component}.
Next we will show: first, how to generate the subproblems; and second, why the computed approximate distances can generate a constant approximation of the original graph metric.

To generate the subproblems, we first sort all edges by weights increasingly, and then use a scan process to partition edges into different ranges.
In the pseudocode we first compute the integer part of the logarithm of each edge to based $n$, and then scan the edge list to find the boundaries between different logarithms.
Then we have the marks of the beginning edges and ending edges in each subproblem.
To compute the logarithms, we use another scan process that computes $n^i$ accordingly and compare the weight of each edge to this value.
Therefore only two types of operations, comparison and multiplication, are used, which are supported in most of the computational models.
This implementation requires edge weights to be no more than $n^{O(m)}$, but if not, we can just add an extra condition in the scan process: if the difference between two consecutive edges in the sorted list is at least $n^2$ times, then instead of aligning the boundary of subproblems using powers of $n$, we leave an empty subproblem between these two edges, and and set the left boundary of the next subproblem to be the larger edge weight.
Since we do not actually use the exact values of the logarithms other than determining the boundaries between subproblems, this change will not affect following computations but bound the amount of subproblems.
This step takes $O(m\log n)$ work.

Since in the subproblems zero-weight edges form vertex components,
we need to keep a data structure to maintain the vertex components in all levels such that: given any specific subproblem,
(1) the vertex component containing any vertex $v$ can be found in $O(\log n)$ time; and
(2) the all vertices in a component can be acquired with the cost linear to the number of vertices in this component.
This data structure is implemented by union-find (tree-based with no path compressions) and requires an extra scanning process as preprocessing.
This preprocessing is similar to run Kruskal's algorithm~\cite{kruskal1956} to compute minimum spanning trees.
We construct this data structure using a scan on the edges in an increasing order of edge weights.
For each edge, if it merges two components, we merge the two trees in union-find by rank, add two directed edges between the roots, and update the root's priority if the component with smaller rank has a higher priority.
Meanwhile the newly added edges also store the loop variable $i$ (line~\ref{line:outer-prepro}).
The vertices in a component in the $i$-th subproblem (line~\ref{line:innterloop}) can be found by traversing from any vertex in the final union-find tree, but ignoring the edges with the stored loop variables larger than $i$.
Finally, given a vertex and a subproblem, the components can be queried by traversing to the root in the union-find tree until reaching the vertex root of this vertex component, and the cost is  $O(\log n)$ per query.
The cost of this step is also $O(m\log n)$.

\myparagraph{Computing \aseq{}s {\normalfont (line~\ref{line:startofmain} -- line~\ref{line:endofmain})}}  After preprocessing, we run the shortest-path algorithm on each subproblem, and the $i$-th subproblem generates the entries in the \aseq{}s in the range of $[n^i, n^{i+1})$.  The search is restricted to edges with weights less than $n^{i+1}$.
Meanwhile the lengths of edges less than $n^{i-1}$ are treated as zero, and the vertices with zero distances can be queried with the union-find tree computed in the preprocessing.
Finally, we solve an extra subproblem that only contains edges with weight less than $n$ to generate the elements in \seq{}s in $[1,n)$.

Observe that in the $i$-th subproblem, if vertex $v$ dominates vertex $u$ with distance at least $n^i$, then (1) $v$ dominates all vertices in the vertex component containing $u$; (2) $v$ has the highest priority in $v$'s vertex component.
We thus can apply the approximate shortest-path described in Section~\ref{sec:algo-prelim} with simple modifications to compute the \seq{}s.
For each vertex component, we only run the shortest-path from the highest priority vertex, whose label is stored in the root node in the union-find tree.
When it dominates another vertex component, then we try to append this vertex to the \seq{} of each vertex in the component.
This process is shown as the function \mf{ComputeDistance} in Algorithm~\ref{algo:linear}.

\myparagraph{Correctness and tree properties}
Lemma~\ref{lem:approx-dis} indicates that optimal tree embeddings can be constructed based on $\hat d_G(u,v)$ as long as $\hat d_G(u,v)$ is computed as $d_u(v)$ by the any single-source shortest-path algorithm that is $\alpha$-\disPreserve{} for any constant $\alpha$.
We now show that the output of Algorithm~\ref{algo:linear}, \aseq{}s $\hat\chi_{\pi}$, is such $\hat d_G(u,v)$ with $\alpha=1/8$.

Theorem~\ref{thm:sp} shows that the computed distances of SSSP algorithm is $1/4$-\disPreserve{}.
Furthermore, the partitioning of the computation into multiple subproblems also leads to an approximation.
The worst approximation lies in the case of a path containing $n-2$ edges with length $n^i-\epsilon$ (for an arbitrary small $\epsilon>0$) and one edge with length $n^{i+1}$.
In this case the approximated distance is $n^{i+1}$, and the actual distance is no more than $(2n-2)n^i$, so the gap is at most a factor of 2.
Clearly this approximation satisfies $\lvert \hat{d}_G(u,v)-\hat{d}_G(u,w)\rvert\le 1/\alpha\cdot d_G(v,w)$.
Combining these two together gives the stated results.

We also need to show that for a pair of vertices $u,v\in V$ and the permutation $\pi$, if $u$ dominates $v$ in $\hat d_G$, the algorithm will add $u$ to ${\hat\chi_{\pi}}^{(v)}$.
Assume $\hat d_G(u,v)\in[n^i,n^{i+1})$, then $u$ dominates $v$ in $\hat d_G$ iff $u$ has the highest priority for the vertex component in the $i$-th subproblem and dominates the vertex component containing $v$.  Then Algorithm~\ref{algo:linear} will correctly add $u$ in ${\hat\chi_{\pi}}^{(v)}$.

\hide{
\begin{lemma}
By replacing the \seq{}s to be the approximate version $\chi'_{\pi}$ computed by Algorithm~\ref{algo:linear}, the algorithm introduced in Section~\ref{sec:tree-constr} can still yield valid FRT trees, i.e.\ for each pair of vertices $u,v\in V$, we have $d(u,v)\le d_T(u,v)$ and $\E[d_T(u,v)]\le O(\log n)\cdot d(u,v)$.
\end{lemma}
\begin{proof}
This proof consists two steps.

In Section~\ref{sec:tree-constr} we show that given the \seq{}s of a graph metric $d_G$, the corresponding FRT trees can be constructed.
Consider another graph metric $d_G'$ that $d_G'(u,v)$ for each pair of vertices $u$ and $v$ is a constant approximation of $d_G(u,v)$, i.e.\ $d_G(u,v)\le d_G'(u,v)<c\cdot d_G(u,v)$ for some constant $c$.
Then if we apply Algorithm~\ref{algo:FRT} on $d_G'$, then we have $d_G'(u,v)\le d_T(u,v)$ and $\E[d_T(u,v)]\le O(\log n)\cdot d_G'(u,v)$,
which further leads to $d_G(u,v)\le d_T(u,v)$ and $\E[d_T(u,v)]\le O(\log n)\cdot d_G(u,v)$.

Then we show that for a graph metric $d_G$, Algorithm~\ref{algo:linear} computes the \seq{}s for a graph metric $d_G'$ which is a constant approximation of $d_G$.
The approximation occurs in two locations.
First, the shortest-path algorithm is approximate.  The result is a 4-approximate distance obeying the triangle inequality.
Second, the partitioning of the computation into multiple subproblems also leads to an approximation.
The worst approximation lies in the case of a path containing $n-2$ edges with length $n^i-\epsilon$ (for an arbitrary small $\epsilon>0$) and one edge with length $n^{i+1}$.
In this case the approximated distance is $n^{i+1}$, and the actual distance is no more than $(2n-2)n^i$, so the gap is at most a factor of 2.
Clearly this approximation obeys triangle inequality as well.
Combining these two locations together, the computed distance for $u,v\in V$, denoted previously as $d_G'(u,v)$, satisfies $d_G(u,v)\le 8d_G'(u,v)\le 8d_G(u,v)$ .

Lastly, we prove that for a pair of vertices $u,v\in V$ and the permutation $\pi$, if $u$ dominates $v$ in $d_G'$, the algorithm will add $u$ to ${\chi'_{\pi}}^{(v)}$.
Assume $d_G'(u,v)\in[n^i,n^{i+1})$, then $u$ dominates $v$ in $d_G'$ iff $u$ has the highest priority for the vertex component in the $i$-th subproblem and dominates the vertex component containing $v$.  Then Algorithm~\ref{algo:linear} will correctly add $u$ in ${\chi'_{\pi}}^{(v)}$.

Altogether, Algorithm~\ref{algo:linear} computes the \seq{}s for a graph metric which is a constant approximation to the original graph metric, and this approximation is independent with $\pi$ and $\beta$.
Then applying the tree construction described in Section~\ref{sec:tree-constr} yields valid FRT trees with the two properties satisfied.
\end{proof}
}

\myparagraph{Time complexity}
Finally we show that Algorithm~\ref{algo:linear} has time complexity $O(m\log n)$  w.h.p.
This indicates that the whole construction process costs $O(m\log n)$,
since the postprocessing to construction of an FRT tree based on \seq{}s has cost $O(n\log n)$ w.h.p.
We have shown that the pre-processing has time complexity $O(m\log n)$.
The remaining parts include the search for approximate shortest paths, and the construction of the \aseq{}s.

We first analyze the cost of shortest-path searches.  For each subproblem, Lemma~\ref{lemma:logn} shows that w.h.p.\ at most $O(\log n)$ vertex components dominate a specific vertex component.  Hence the overall operation number of \mf{Decrease-Key} and \mf{Extract-Min} is $O(m'\log n)$ where $m'$ is the number of edges in this subproblem.  Since each edge is related to at most two subproblems, the overall cost of shortest-path search is $O(m\log n)$ w.h.p.

Next we provide the analysis of the construction of \seq{}s.
We have shown that Algorithm~\ref{algo:linear} actually computes the \seq{}s of an approximate graph metric, so the overall number of elements in $\hat\chi_{\pi}$ is $O(n\log n)$ w.h.p.
Since the vertex components can be retrieved using the union-find structure, the cost to find each vertex is $O(1)$ asymptotically.
Note that in the $i$-th subproblem, if $\hat d_G(u,v)\ge n^i$, $u$ dominates $v$ if and only if $u$ dominates all vertices in $v$'s component.
Therefore $u$ is added to all ${\hat\chi_{\pi}}^{(v')}$ for $v'$ in this component.
Since the subproblems are processed in decreasing order on distances (line~\ref{line:startofmain}), there will be no multiple insertions for a vertex in an \aseq{}, nor further removal for an inserted element in the sequence.
Lemma~\ref{lemma:sigma} bounds the total number of vertices in all components $S$ by $O(n \log n)$ w.h.p.\ and each insertion takes $O(1)$ time, so the overall cost to maintain the \aseq{s} is $O(n \log n)$ w.h.p.
In total, the overall time complexity to construct the \aseq{}s is $O(m \log n)$ w.h.p.

\begin{algorithm}[!th]
\caption{Construction of the \aseq{}s}
\label{algo:linear}
\KwIn{A weighted graph $G=(V,E)$ and a random permutation $\pi$.}
\KwOut{Approximate \seq{}s $\chi'_{\pi}$ and the set of associated distances $d$.}
    \vspace{0.5em}
    Initialize each vertex as a singleton component, i.e.\ $f(v)\gets v$\\
    Sort all edges based on their weights so that $d(e_1)\le d(e_2)\le\cdots\le d(e_m)$\\
    $d\leftarrow \varnothing$\\
    Computer the logarithm of each edge: $\pw{}_i\gets \lfloor\log_n{d(e_i)}\rfloor$\\
    Let $\stt{}_i$ be the first edge that $\pw{}_i=i$\\
    \For{$i\gets 0$ to $\lfloor\log_n{d(e_m)}\rfloor$\label{line:outer-prepro}} {
            \For{$j\gets \stt{}_i$ to $\stt{}_{i+2}-1$} {
                Mark the corresponding components for edge $e_j$
            }
            \For{$j\gets \stt{}_i$ to $\stt{}_{i+1}-1$} {
                \If{$e_j$ connects two different components} {
                    Merge the smaller component into the larger one\\
                    The new component has the priority to be the higher among the previous two\\
                }
            }
    }\label{line:endofprepro}
    \For{$i\gets \lfloor\log_n{d(e_m)}\rfloor$ downto 0\label{line:startofmain}} {
        \mf{ComputeDistance}$(\{e_{\script{\stt{}}_i},\cdots,e_{\script{\stt{}}_{i+2}-1}\}, d(e_{\script{\stt{}}_{i+1}}))$\\
    }
    \mf{ComputeDistance}$(\{e_1,\cdots,e_{i-1}\},1)$\label{line:endofmain}\\
\hide{
    $i''\gets 1$\\
    \For{$i\leftarrow 1$ to $m$} {
        \lIf{$d(e_i)>n$} {break}
    }
    $i'\gets i, \en{}_i\gets -1$\\

    \For{$i\gets i'$ to $m+1$} {
        \If{$i=m+1$ or $d(e_i)>n^2\cdot d(e_{i''})$} {
            $\stt{}_i\gets i'', \en{}_i\gets i'$\\
            \For{$j\gets i''$ to $i-1$} {
                Mark the corresponding components for edge $e_j$
            }
            \For{$j\gets i''$ to $i'-1$} {
                \If{$e_j$ connects two different components} {
                    Merge the smaller component into the larger one\\
                    The new component has the priority to be the higher among the previous two\\
                }
            }
            $i''\gets i', i'\gets i$
        }
    }
    $i\gets m+1$\label{line:endofprepro}\\
    \While {$\en{}_i>0$} {
        \mf{ComputeDistance}$(\{e_{\script{\stt{}}_i},\cdots,e_{i-1}\},n\cdot d(e_{\script{\stt{}}_i}))$\\
        $i\gets \en{}_i$
    }
    \mf{ComputeDistance}$(\{e_1,\cdots,e_{i-1}\},1)$\\
}
    \Return{$\hat\chi_{\pi}$, $d$}\\
    \medskip
    \Fn{\mf{ComputeDistance}(edgeset $E'$, range $r$)} {
        Let $C$ be the set of vertex components that contains the endpoints associated to $E'$\\
        Sort the components based on the priority, i.e.\ $\pi(V'_1)<\cdots<\pi(V'_{|C|})$ for all $V'_i\in C$\\
        Build adjacency list by sort the edges $E'$ based on the component endpoints\\
        \lForEach{$V'\in C$} {
        $\delta(V')\leftarrow+\infty$
        }
        \For{$i\leftarrow 1$ to $|C|$}
        {
            Let $p$ be the vertex with the highest priority in $V'_i$\\
            Apply approximate SSSP algorithm to compute $S=\{S'\in C\mid d(p,S')<\delta(S')\}$\\
            \For{$S'\in S$} {
                $\delta(S')\gets d(p,S')$\\
                \If {$\delta(S')>r$} {
                    \For {$u'\in S'$\label{line:innterloop}} {
                        Try to append $p$ to ${\hat\chi_{\pi}}^{(u')}$\\
                        If successful, $d\gets d\cup \{(u',p)\rightarrow 8\cdot\delta(S')\}$\label{line:distance}
                    }
                }
            }
        }
    }
    \Return{$\hat\chi_{\pi}$, $d$}
\end{algorithm}
\section{Asymptotically Tight Ramsey Partitions Based on FRT Trees}
\label{sec:frt-ramsey}

In this section we show that the probability distribution over FRT trees is an asymptotically tight Ramsey Partition. Therefore, FRT trees can be used to construct approximate distance oracle with size $O(n^{1+1/k})$, $O(k)$ stretch and $O(1)$ query time using the algorithm presented in~\cite{Ramsey}. The construction time of this oracle is $O(n^{1/k}(m\log n+n\log^2 n))$, which is not only faster than the best known algorithm~\cite{FastCKR} with $O(n^{1/k}m\log^3 n)$ time complexity for construction, but also much simpler.
More importantly, our new algorithm provides an $18.5k$-approximation, smaller than the bounds $128k$ in~\cite{Ramsey} and $33k$ in~\cite{naor2012}.

Recall the definition of Ramsey partitions: given a metric space $(X,d_X)$, an $(\alpha, \gamma)$ Ramsey partition is the probability distribution over partition trees $\{\mathcal{P}_k\}_{k=0}^\infty$ of $X$, such that:
$$\Pr\left[\forall i\in \mathbb{N},B_X\left(x,\alpha\cdot c^{-i}\Delta\right)\subseteq \mathcal{P}_i(x)\right]\geq|X|^{-\gamma}.$$
\begin{theorem}
\label{thm:FRTRAM}
The probability distribution over FRT trees is an asymptotically tight Ramsey Partition with $\alpha=\Omega(\gamma)$ (shown in the appendix) and fixed $c=2$. More precisely, 
for every $x\in X$,
$$\Pr\left[\forall i\in \mathbb{N},B_X\left(x, \left(1-2^{-1/2a}\right){2^{-i}\Delta}\right)\subseteq \mathcal{P}_i(x)\right]\geq{1\over 2}|X|^{-{2\over a}}$$
for any positive integer $a>1$.
\end{theorem}
Note that the $1/2$ on the right side of the inequality only leads a doubling of the overall trees needed for the Ramsey-based ADO, but does not affect
asymptotic bounds.


To prove Theorem~\ref{thm:FRTRAM}, we first prove some required lemmas.

\begin{lemma}
\label{lemma:bucket}
Given $n-1$ arbitrarily chosen integers $v_1,v_2,\cdots,v_{n-1}$ in
$[a]$ ($a\ge 2$), let $S$ be a set in which each $v_i$ is selected
independently at random with probability $p_i\leq {1/(i+1)}$.
 Given an integer $v$ chosen uniformly at random from $[a]$,
 we have that \[\Pr\left[\{s = v : s \in S\} = \varnothing\right] \geq {\epsilon}\cdot
 n^{-1/(a(1-\epsilon))}\]
 for any $\epsilon\in
 [a-1]/a$\footnote{$\epsilon\in [a-1]/a$ means $\epsilon=i/a$ where
   $i\in[a-1]$.}.
\end{lemma}

\begin{proof}
Construct $a$ buckets so that each value is in the $v_i$-th bucket.
Let $s_j=\ln \prod_{v_i\in{\rm{bucket}}_j}(1-p_i)=\sum_{v_i\in {\rm{bucket}_j}}\ln(1-p_i)$.
With the fact that $\prod_{i=2}^{n}(1-1/i)=1/n$,
we have $\sum_{j=1}^{a}s_j=\sum_{i=1}^{n}\ln(1-p_i)\ge \sum_{i=1}^{n-1}\ln (i/(i+1))=-\ln n$.
Therefore, there exists at least $\epsilon a$ buckets such that their $s_j$ are at least $-1/(a(1-\epsilon))\cdot \ln n$.
For these buckets, the probability that none of the points in a specific bucket are selected is $\prod_{v_i\in{\rm{bucket}_j}}(1-p_i)=e^{s_j}\ge n^{-1/(a(1-\epsilon))}$.
Hence the probability that a randomly chosen bucket is empty is at least $\epsilon\cdot n^{-1/(a(1-\epsilon))}$.
\end{proof}

\hide{
\begin{proof}
Construct $a$ buckets so that each value is in the $v_i$-th bucket. Let $s_i$ be the sum of the probabilities for the values in the $i$-th bucket.

Since $\sum_{i=1}^{n}p_i=\sum_{i=1}^{a}s_i<\ln n$, at least $a/2$ buckets that $s_i$(s) are less than $(2\ln n)/a$. For these buckets, the probability that none of the points in a specific bucket are selected is $\prod_{v_i\in{\footnotesize\tb{bucket}}}(1-p_i)>\exp(1.4\cdot (2\ln n)/a)=n^{-{2.8\over a}}$. Therefore on average, the probability that a chosen bucket is empty is at least ${1\over 2}n^{-{2.8\over a}}$.
\end{proof}}

\begin{lemma}
\label{lemma:range}
Given $n-1$ points $v_1,v_2,\cdots,v_{n-1}$ in $[0,1)$, select each
independently with probability $p_i\leq {1/(i+1)}$.
For a uniformly randomly-picked $b$ in $[0,1)$, the probability that no point $v_i$ in the range $[b-1/2a,b+1/2a) \newmod 1$\footnote{$\newmod 1$ here means to wrap the interval into the range $[0,1)$. For example, $[0.8,1.2)\newmod 1=[0.8,1)\cup[0,0.2)$.} is selected is at least ${\epsilon}\cdot n^{-1/(a(1-\epsilon))}$ for any positive integer $a\ge 2$ and $\epsilon\in [a-1]/a$.
\end{lemma}

\begin{proof}
Given $b$, we construct $a$ buckets where the $i$-th bucket contains the points that fall into the range $[b+(i-1/2)/a,b+(i+1/2)/a)\newmod 1$. Since any $b'=(b+i/a)\newmod 1~(i\in[a])$ will create the same buckets and $b$ is uniformly distributed in $[0,1)$, by applying Lemma~\ref{lemma:bucket}, the probability that a chosen bucket is empty is at least ${\epsilon}\cdot n^{-1/(a(1-\epsilon))}$ on average.
\end{proof}

In Lemma~\ref{lemma:range}, $b$ will be related to $\beta$ in the FRT tree, and $a$ will represent the padded ratio in Theorem~\ref{thm:FRTRAM}.

\begin{proof}[Proof of Theorem~\ref{thm:FRTRAM}]
In this proof, we focus on one specific element $x\in X$.

Let $d_1,d_2,\cdots,d_{|X|-1}$ be the distances from $x$ to all the other elements. We transform these distances to the log-scale and ignore the integer parts by denoting $\delta(d_i)=(\log_2 d_i) \newmod 1$.
Therefore, the radii of the hierarchical partitions $\{\beta\Delta,\beta\Delta/2,\beta\Delta/4, \dots\}$ will be exactly the same as $\beta'=\log_2(\beta\Delta)\newmod 1=\log_2 \beta$ after this transformation, and $\beta'$ follows the uniform distribution on $[0,1]$ according to the FRT algorithm.

WLOG, assume the distances are in increasing order, i.e. $0<d_1\leq d_2\leq \cdots\leq d_{|X|-1}$. Due to the property of a random permutation, the probability for each element to be in the \seq~of $x$ is $p_1={1\over 2},p_2={1\over 3},\cdots,p_{n-1}={1\over n}$.

Now we compute the probability that none of the distances from $x$ to the elements in the \seq~of $x$ are in the padded region near the boundary. The radii of the partitions are $2^{-i}\beta\Delta$, so the padded region stated in Theorem~\ref{thm:FRTRAM} is $\bigcup_{i=0}^{+\infty}\left[\left(1-\left(1-2^{-1/2a}\right)\right)\cdot 2^{-i}\beta\Delta,\left(1+\left(1-2^{-1/2a}\right)\right)\cdot 2^{-i}\beta\Delta\right]$. These intervals after transformation are all in $[\beta'-1/2a,\beta'+1/2a)\newmod 1$. 
The probability that none of the elements in the interval $[\beta'-1/2a,\beta'+1/2a)\newmod 1$ is in the \seq~of $x$, is equivalent to elements not
being selected in the range in Lemma~\ref{lemma:range}.   This gives a probability of at least ${1\over 2}|X|^{-{2/a}}$ if $\epsilon$ is set to be $1/2$ (or $\epsilon=(a-1)/2a$ for odd integer $a$).

Since all of $x$'s ancestors are in the \seq~of $x$, the probability that none of them are in that padded range is also at least ${1\over 2}|X|^{-{2/a}}$.
\end{proof}
\begin{remark}
\emph{
There are two major differences between this construction (based on FRT trees) and the original construction in~\cite{Ramsey} (based on hierarchical CKR partitions). First, there is only one random permutation for all levels in one FRT tree, but the hierarchical CKR partitions requires different random permutations in different levels. Second, the radii for the partitions decrease exactly by a half between two consecutive levels in FRT trees, but the radii in hierarchical CKR partitions is randomly generated in every level.}

\emph{Theorem~\ref{thm:FRTRAM} shows that the extra randomness to generate multiple random permutation and radii is unnecessary to construct Ramsey partitions. Moreover, since the radii for the partitions decrease by a half between two consecutive levels (instead of at most 16 in the original construction), the stretch can be reduced to $18.5k$ as opposed to $128k$, which is shown in appendix~\ref{sec:approx}.  Lastly, the proof is also very different and arguably simpler.}
\end{remark}

\begin{corollary}
\label{prob4}
With Theorem~\ref{thm:FRTRAM}, we can accelerate the time to construct the Ramsey-partitions-based Approximate Distance Oracle in~\cite{Ramsey} to
$$O\left(n^{1/k}(m+n \log n)\log n\right)$$
on a graph with $n$ vertices and $m$ edges, improving the stretch to $18.5k$,
while maintaining the same storage space and constant query time.
\end{corollary}

This can be achieved by replacing the original hierarchical partition trees in the distance oracles by FRT trees (and some other trivial changes).
The construction time can further reduce to $O\left(n^{1/k}m\log n\right)$ using the algorithm introduced in Section~\ref{sec:eff-constr-app-seq} while the oracle still has a constant stretch factor.
Accordingly, the complexity to construct
Christian Wulff-Nilsen's Distance Oracles~\cite{WN} and Shiri Chechik's Distance Oracles~\cite{SC} can be reduced to
$$O\left(kmn^{1/k}+kn^{1+1/k}\log n+n^{1/ck}m\log n\right)$$
since they all use Mendel and Naor's Distance Oracle to obtain an initial distance estimation.
The acceleration is from two places: first, the FRT tree construction is faster; second, FRT trees provide better approximation bound, so the $c$ in the exponent becomes smaller.

\hide{
\begin{proof}
For simplicity, define $P(x)=\Pr\left[\forall k\in \mathbb{N},B_X\left(x, \left(1-2^{-1/4\lfloor a\rfloor}\right){2^{-k}\diam(X)}\right)\subseteq \mathcal{P}_k(x)\right]$ for $x\in X$.

We assume that an adversary is able to pick the any distances to other elements but not the parameters including the random permutation $\pi$ and $\beta$ to generate FRT trees, and we will show that the worst case follows the theorem.

This proof focuses on one specific element $x$ in the metric space. We first liberalize the restriction that a metric space must be symmetric to strengthen the adversary, such that the adversary can choose the distances of each element without interfering its actions with other elements.

Let $d_1,d_2,\cdots,d_{|X|-1}$ be the distances from $x$ to all the other elements. We transform these distances to the log-scale and ignore the integer parts by denoting $\delta(d_i)=\{\log_2 d_i\}$ where $\{\cdot\}$ is the Sawtooth function. The benefit for doing this is because the radii of hierarchical partitions after the transformation will be exactly the same. We want to show that when the random permutation is generated and the $\beta$ is randomly picked (uniformly random after transformation based on the FRT algorithm), the probability that none of $x$'s ancestors (elements in \cps~of $x$) is closer than a ratio of $2^{-1/4\lfloor a\rfloor}$ of the radius to the boundary of the partition in the corresponding level is at least ${1\over 2}|X|^{-{2.8\over a}}$. For simplicity, we change the ancestors to all elements in \seq~of $x$ in this argument. Since the elements in \cps~of $x$ are a subset of the elements in \seq~of $x$, the overall probability will not increase.

WLOG, assume the adversary picks the distances in increasing order, i.e. $0<d_1\leq d_2\leq \cdots\leq d_{|X|-1}$. Due to the property of a random permutation, the probability for each element to be in the \seq~is $p_1={1\over 2},p_2={1\over 3},\cdots,p_{n-1}={1\over n}$. (Probably some explanation here.) (Therefore, $P(x)=\min_{k=0}^{\infty}{\left|{1-\beta2^{-k}\diam(X)}\right|}$ donno how to write this....)

Now we partition the range of $\{\delta_i\}$ into $a$ buckets, and the $j$-th bucket corresponds to the range of $\left[(j-1)/a,j/a\right)$. In Theorem~\ref{thm:FRTRAM} we have $\alpha = 1-2^{-1/4\lfloor a\rfloor}$. Since $\log_2(1-\alpha)=\log_2(2^{-1/4\lfloor a\rfloor})=-1/4\lfloor a\rfloor$ and $\log_2(1+\alpha)<\log_2(1/(1-\alpha))=1/4\lfloor a\rfloor$, if $\delta(\beta\cdot2^{-k}\diam(X))$ falls into the region of $\left[(j-0.75)/a,(j+0.75)/a\right]$, then the elements in interval $[(1-\alpha)\beta\cdot2^{-k}\diam(X),(1+\alpha)\beta\cdot2^{-k}\diam(X)]$ will only be in bucket $j$.

Now we analysis the probability that none of the elements in bucket $j$ is in \seq~of $x$ no matter how the adversary chooses the distances, which is equivalent to put the $i$-th element into one arbitrary bucket with the probability of $p_i=1/(i+1)$ to be in the \seq~of $x$. Since $\sum_{i=1}^{|X|-1}p_i<\ln|X|$, we know that at least $a/2$ buckets have elements that the sum of $p_i$ in each bucket is no more than $2\ln|X|/a$. Therefore, the probability that none of the elements appears in \seq~of $x$ is $\prod_{i}(1-p_i)>\exp(1.4\cdot 2\ln|X|/a)=|X|^{-{2.8\over a}}$. Since there might be at most $a/2$ buckets that contain elements that the sum of $p_i$ in each bucket is more than $2\ln|X|/a$, a uniformly randomly chosen bucket that has no elements in \seq~of $x$ is at least ${1\over 2}|X|^{-{2.8\over a}}$.
\end{proof}
}

\section{Conclusion}

In this paper we described a simple and efficient algorithm to construct FRT embeddings on graphs using $O(m \log n)$ time, which is based on a novel linear-time algorithm for single-source shortest-paths proposed in this paper.
We also studied the distance preserving properties on FRT embeddings, and proved that FRT trees are asymptotic optimal Ramsey partitions.
Lastly, we used the FRT trees to construct a simple and programming-friendly distance oracle, which surprisingly shows good approximation on all of our testing cases.


\bibliographystyle{plain}
\bibliography{main}

\begin{thebibliography}{10}

\bibitem{agarwal2013}
Rachit Agarwal and Philip Godfrey.
\newblock Distance oracles for stretch less than 2.
\newblock In {\em Proceedings of ACM-SIAM Symposium on Discrete Algorithms
  (SODA)}, pages 526--538, 2013.

\bibitem{alon1995graph}
Noga Alon, Richard~M Karp, David Peleg, and Douglas West.
\newblock A graph-theoretic game and its application to the k-server problem.
\newblock {\em SIAM Journal on Computing}, 24(1):78--100, 1995.

\bibitem{Bartal96}
Yair Bartal.
\newblock Probabilistic approximation of metric spaces and its algorithmic
  applications.
\newblock In {\em In Proceedings of IEEE Foundations of Computer Science
  (FOCS)}, pages 184--193, 1996.

\bibitem{Bartal98}
Yair Bartal.
\newblock On approximating arbitrary metrices by tree metrics.
\newblock In {\em Proceedings of ACM Symposium on Theory of Computing (STOC)},
  pages 161--168. ACM, 1998.

\bibitem{baswana2006}
Surender Baswana and Telikepalli Kavitha.
\newblock Faster algorithms for approximate distance oracles and all-pairs
  small stretch paths.
\newblock In {\em In Proceedings of IEEE Symposium on Foundations of Computer
  Science (FOCS)}, pages 591--602, 2006.

\bibitem{blelloch2016parallelism}
Guy~E Blelloch, Yan Gu, Julian Shun, and Yihan Sun.
\newblock Parallelism in randomized incremental algorithms.
\newblock In {\em In Proceedings of ACM Symposium on Parallelism in Algorithms
  and Architectures (SPAA)}, pages 467--478, 2016.

\bibitem{Kanat}
Guy~E Blelloch, Anupam Gupta, and Kanat Tangwongsan.
\newblock Parallel probabilistic tree embeddings, k-median, and buy-at-bulk
  network design.
\newblock In {\em In Proceedings of ACM symposium on Parallelism in Algorithms
  and Architectures (SPAA)}, pages 205--213, 2012.

\bibitem{CKR}
Gruia Calinescu, Howard Karloff, and Yuval Rabani.
\newblock Approximation algorithms for the 0-extension problem.
\newblock {\em SIAM Journal on Computing}, 34(2):358--372, 2005.

\bibitem{SC}
Shiri Chechik.
\newblock Approximate distance oracles with constant query time.
\newblock In {\em In Proceedings of ACM Symposium on Theory of Computing
  (STOC)}, pages 654--663, 2014.

\bibitem{chechik15}
Shiri Chechik.
\newblock Approximate distance oracles with improved bounds.
\newblock In {\em In Proceedings of ACM on Symposium on Theory of Computing
  (STOC)}, pages 1--10, 2015.

\bibitem{cohen1997}
Edith Cohen.
\newblock Size-estimation framework with applications to transitive closure and
  reachability.
\newblock {\em Journal of Computer and System Sciences}, 55(3):441--453, 1997.

\bibitem{cohen2000}
Edith Cohen.
\newblock Polylog-time and near-linear work approximation scheme for undirected
  shortest paths.
\newblock {\em Journal of the ACM (JACM)}, 47(1):132--166, 2000.

\bibitem{Cohen2014}
Michael~B. Cohen, Rasmus Kyng, Gary~L. Miller, Jakub~W. Pachocki, Richard Peng,
  Anup~B. Rao, and Shen~Chen Xu.
\newblock Solving {SDD} linear systems in nearly {$m\log^{1/2}n$} time.
\newblock In {\em Proceedings of ACM Symposium on Theory of Computing (STOC)},
  pages 343--352, 2014.

\bibitem{Dijk}
Edsger~W Dijkstra.
\newblock A note on two problems in connexion with graphs.
\newblock {\em Numerische mathematik}, 1(1):269--271, 1959.

\bibitem{elkin2015}
Michael Elkin and Seth Pettie.
\newblock A linear-size logarithmic stretch path-reporting distance oracle for
  general graphs.
\newblock In {\em Proceedings of ACM-SIAM Symposium on Discrete Algorithms
  (SODA)}, pages 805--821, 2015.

\bibitem{FRT}
Jittat Fakcharoenphol, Satish Rao, and Kunal Talwar.
\newblock A tight bound on approximating arbitrary metrics by tree metrics.
\newblock {\em Journal of Computer and System Sciences (JCSS)}, 69(3):485--497,
  2004.

\bibitem{Fib}
Michael~L Fredman and Robert~Endre Tarjan.
\newblock Fibonacci heaps and their uses in improved network optimization
  algorithms.
\newblock {\em Journal of the ACM (JACM)}, 34(3):596--615, 1987.

\bibitem{Friedrichs2016}
Stephan Friedrichs and Christoph Lenzen.
\newblock Parallel metric tree embedding based on an algebraic view on
  moore-bellman-ford.
\newblock In {\em Proceedings of ACM Symposium on Parallelism in Algorithms and
  Architectures (SPAA)}, pages 455--466, 2016.

\bibitem{gabow1985scaling}
Harold~N Gabow.
\newblock Scaling algorithms for network problems.
\newblock {\em Journal of Computer and System Sciences}, 31(2):148--168, 1985.

\bibitem{ghaffari2014near}
Mohsen Ghaffari and Christoph Lenzen.
\newblock Near-optimal distributed tree embedding.
\newblock In {\em International Symposium on Distributed Computing}, pages
  197--211. Springer, 2014.

\bibitem{karger2002}
David~R Karger and Matthias Ruhl.
\newblock Finding nearest neighbors in growth-restricted metrics.
\newblock In {\em In Proceedings of ACM symposium on Theory of Computing
  (STOC)}, pages 741--750. ACM, 2002.

\bibitem{Khan2008}
Maleq Khan, Fabian Kuhn, Dahlia Malkhi, Gopal Pandurangan, and Kunal Talwar.
\newblock Efficient distributed approximation algorithms via probabilistic tree
  embeddings.
\newblock {\em Distributed Computing}, 25(3):189--205, 2012.

\bibitem{klein1997}
Philip~N Klein and Sairam Subramanian.
\newblock A randomized parallel algorithm for single-source shortest paths.
\newblock {\em Journal of Algorithms}, 25(2):205--220, 1997.

\bibitem{kruskal1956}
Joseph~B Kruskal.
\newblock On the shortest spanning subtree of a graph and the traveling
  salesman problem.
\newblock {\em Proceedings of the American Mathematical society}, 7(1):48--50,
  1956.

\bibitem{Ramsey}
Manor Mendel and Assaf Naor.
\newblock Ramsey partitions and proximity data structures.
\newblock In {\em In Proceedings of IEEE Symposium on Foundations of Computer
  Science (FOCS)}, pages 109--118, 2006.

\bibitem{FastCKR}
Manor Mendel and Chaya Schwob.
\newblock Fast {CKR} partitions of sparse graphs.
\newblock {\em Chicago Journal of Theoretical Computer Science}, (2):1--18,
  2009.

\bibitem{miller2015}
Gary~L. Miller, Richard Peng, Adrian Vladu, and Shen~Chen Xu.
\newblock Improved parallel algorithms for spanners and hopsets.
\newblock In {\em In Proceedings of {ACM} Symposium on Parallelism in
  Algorithms and Architectures (SPAA)}, pages 192--201.

\bibitem{Morrison68}
Donald~R. Morrison.
\newblock Patricia---practical algorithm to retrieve information coded in
  alphanumeric.
\newblock {\em Journal of the ACM (JACM)}, 15(4):514--534, October 1968.

\bibitem{naor2012}
Assaf Naor and Terence Tao.
\newblock Scale-oblivious metric fragmentation and the nonlinear {D}voretzky
  theorem.
\newblock {\em Israel Journal of Mathematics}, 192(1):489--504, 2012.

\bibitem{pettie2005}
Seth Pettie and Vijaya Ramachandran.
\newblock A shortest path algorithm for real-weighted undirected graphs.
\newblock {\em {SIAM} J. Comput.}, 34(6):1398--1431, 2005.

\bibitem{Racke08}
Harald R{\"a}cke.
\newblock Optimal hierarchical decompositions for congestion minimization in
  networks.
\newblock In {\em Proceedings of ACM Symposium on Theory of Computing (STOC)},
  pages 255--264, 2008.

\bibitem{Seidel93}
Raimund Seidel.
\newblock {\em Backwards analysis of randomized geometric algorithms}.
\newblock Springer, 1993.

\bibitem{thorup1997}
Mikkel Thorup.
\newblock Undirected single source shortest paths in linear time.
\newblock In {\em In Proceedings of IEEE Symposium on Foundations of Computer
  Science (FOCS)}, pages 12--21, 1997.

\bibitem{TZ}
Mikkel Thorup and Uri Zwick.
\newblock Approximate distance oracles.
\newblock {\em Journal of the ACM (JACM)}, 52(1):1--24, 2005.

\bibitem{WN}
Christian Wulff-Nilsen.
\newblock Approximate distance oracles with improved query time.
\newblock In {\em Proceedings of ACM-SIAM Symposium on Discrete Algorithms
  (SODA)}, pages 539--549, 2013.

\end{thebibliography}
\appendix
\section{Proof for Lemma \ref{lemma:sigma}}
\label{app:lemma3}

\begin{proof}[Proof of Lemma~\ref{lemma:sigma}]
For a fixed point $x\in X$, we first sort all of the vertices by their distances to $x$, and hence $0=d(x,x_1'=x)<d(x,x_2')\le d(x,x_3')\le \cdots \le d(x,x_n')$. Let $y_i=\min\{\pi(x_j')\,|\,j\leq i\}$. By definition, $x_i'$ is in the \seq~of $x$ if and only if $\pi(x_i')<\pi(x_j')$ for all $j<i$, which is also equivalent to $y_i<y_{i-1}$ for $i>1$. Thus, $\left|\chi_{\pi}^{(x)}\right|$ is the number of different elements in $y_i, 1\le i \le n$. Since $\pi$ is a random permutation and independent to the distances to $x$, by applying Lemma~\ref{lemma:logn}, there exist $O(\log n)$ elements in $\chi_{\pi}^{(x)}$ w.h.p.

Therefore, the total elements in \seq{}s of $G$ are $O(n\log n)$ w.h.p.\ simply by summing up \seq{}s for all vertices.
\end{proof}

\section{\Seq{} Construction based on Dijkstra's Algorithm}
\label{app:LE-list}

\begin{algorithm}[!tp]
\caption{The \seq{}s construction}
\label{algo:seq}
\KwIn{A weighted graph $G=(V,E)$ and a random permutation $\pi$.}
\KwOut{\Seq{}s $\chi_{\pi}$ and the set of associated distances $d$.}
    \vspace{0.5em}
    \lForEach{$v\in V$} {
    $\delta(v)\leftarrow+\infty$
    }
    $d\leftarrow \varnothing$\\
    \For{$i\leftarrow 1$ to $|V|$}
    {
        $u\leftarrow \pi^{-1}(i)$\\
        $Q\leftarrow \{u\}$\\
        $\delta(u)\leftarrow 0$\\
        \While{$Q\ne\varnothing$} {
            Extract $v\in Q$ with minimal $\delta(v)$\\
            \nllabel{extract}
            $\chi_{\pi}^{(v)}\leftarrow \chi_{\pi}^{(v)}+u$~~~~// add $u$ to $v$'s \seq\\
            $d\leftarrow d\cup\{(u,v)\rightarrow\delta(v)\}$\\
            \ForEach{$w: (v,w,l_{v,w})\in E$}
            {
                \If{$\delta(w)>\delta(v)+l_{v,w}$ \nllabel{check}}
                    {
                    $\delta(w)\leftarrow\delta(v)+l_{v,w}$~~~~// $l_{v,w}$ is the edge length\\
                    \nllabel{decrease}
                    $Q\leftarrow Q\cup \{w\}$\\
                    }
            }
        }
     }
    \Return{$\chi_{\pi}$, $d$}
\end{algorithm}

Here we review the algorithm proposed by Cohen~\cite{cohen1997} to compute \seq{}, which is basically a variation of Dijkstra's algorithm. The \seq{}s $\chi_{\pi}$ can be trivially constructed by running SSSP from each vertex, but this takes $O(n(m+n\log n))$ time.
An important observation is that:
for each vertex $u$, if its distance to another vertex $v$ is larger than $v$'s distance to a higher priority vertex $\bar{u}$, then it is not necessary to explore $v$ when running the SSSP from $u$.
This is because $\bar{u}$ is closer to $v$ and has a higher priority.
Similar ideas can be found in related works~\cite{TZ,FastCKR}.

Inspired by this observation, we run Dijkstra algorithm with the
source vertices in increasing order in $\pi$, but for each source
vertex $u$, we only explore the vertices that are closer to $u$ than
any other previous source vertices.  The algorithm first starts the
SSSP from $\pi^{-1}(1)$ and adds this vertex to the \seq~of all
vertices, then it starts the SSSP from $\pi^{-1}(2)$ but only adds it
to the \seq~of the vertices that are closer to $\pi^{-1}(2)$ than
$\pi^{-1}(1)$, and this repeats for $n$ times.  In the algorithm,
$\delta(u)$ in the $i$-th round is the shortest distance of $u$ to any
of the first $i$ vertices, so that the SSSP is restricted in
each round to vertices for which $\delta(u)$ is updated.
Since $\delta(u)$ is not initialized in each round, the overall time
complexity is largely decreased.

\begin{lemma}
Given a graph metric $(X,d_X)$ and a random permutation $\pi$,
the \seq{}s $\chi_{\pi}$ can be constructed in $O(m\log n+n\log^2 n)$
time w.h.p.
\end{lemma}
\begin{proof}
Lemma~\ref{lemma:sigma} indicates that w.h.p.\ the total
number of delete-min operations (line \ref{extract}) is $|\chi_{\pi}|=O(n\log n)$ and the overall number of  decrease-key operations (line \ref{decrease}) is $\sum_{u}\deg(u)\cdot\left|\chi_{\pi}^{(u)}\right|=O(m\log n)$.
Thus, if the priority queue $Q$ in Algorithm~\ref{algo:seq} is implemented using Fibonacci heap, the overall time complexity for Algorithm~\ref{algo:seq} is $O(m\log n+n\log^2 n)$ w.h.p.
\end{proof}

Lastly, we show the correctness of the \seq{}s construction algorithm.

\begin{lemma}
\label{lemma:correctness-dom-seq}
Given the random permutation $\pi$, and the graph $G=(V,E)$, the dominance sequence of a specific vertex $v \in V$ will be stored in $\chi_{\pi}^{(v)}$ at the end of Algorithm \ref{algo:seq}, and consequently the distances between every vertex and its dominating vertices are in $d$.
\end{lemma}

\begin{proof}
First we prove that all vertices that dominate $v$ are in $\chi_{\pi}^{(v)}$.

Assuming that $u\in V$ is a vertex and dominates $v$, we show that $u\in \chi_{\pi}^{(v)}$. Let $(w_1=u,w_2,\dots,w_{k-1},w_k=v)$ be the shortest path from $u$ to $v$. Since $u$ dominates $v$, $\pi(u)<\pi(w_i)$ holds for all $i=2,3,\dots, k$.

We first prove that $u$ dominates all the vertices on the shortest path. Assume to the contrary that there exists at least one vertex $w_i$ that is not dominated by $u$, which means that at least another vertex $u^*$ holds both $\pi(u^*)<\pi(u)$ and $d(u^*,w_i)\leq d(u,w_i)$. This indicates that the distance from $u^*$ to $v$ is at most $d(u^*,w_i)+\sum_{j=i}^{k-1}d(w_j,w_{j+1})\leq\sum_{j=1}^{k-1}d(w_j,w_{j+1})=d(u,v)$. Combining with the other assumption $\pi(u^*)<\pi(u)$, it directly leads to a contradiction that $u^*$ dominates $v$, since $u^*$ is closer to $v$ and owns a higher priority.


Then we prove that when the outermost for-loop is at vertex $u$, vertex $v$ will be extracted (in line \ref{extract}) from $Q$. Then $u$ will be added into $\chi_{\pi}^{(v)}$ and $\delta(v)$ will be updated by $d(u,v)$ by induction.
Initially $w_1=u$ is added to $\chi_{\pi}^{(w_1)}$ (and changes $\delta(w_1)$ to be $d(u,u)=0$), which is the base case.
On the inductive step, we show that as long as $u$ is attached to $\chi_{\pi}^{(w_{i-1})}$ in line \ref{extract}, all $w_{i}$ will be inserted into $Q$ later in line \ref{decrease}, and finally $u$ will be added to $\chi_{\pi}^{(w_{i})}$.
From the previous conclusion, we know that $u$ dominates $w_i$ via the shortest path $(w_1,w_2,\dots,w_i)$, so that $\delta(w_i)>d(u,w_{i-1})+d(w_{i-1},w_i)=\delta(w_{i-1})+d(w_{i-1},w_i)$, which guarantees that the checking in line \ref{check} for edge $(w_i,w_{i+1})$ will be successful. After that, $w_{i}$ will be inserted into $Q$ and $u$ will be attached to $\chi_\pi^{(v)}$ later, and $\delta(w_i)$ will be updated to $d(u,w_{i-1})+c_{w_{i-1},w_i}$, which is the new $d(u,w_i)$.

We now show that all vertices in $\chi_{\pi}^{(v)}$ dominate $v$. Assume to the contrary that at the end of the algorithm, $\exists u \in \chi_{\pi}^{(v)}$ and $u$ does not dominate $v$. This means there exists $u^*$ which satisfies both prior to $u$ and $d(v,u^*)<d(v,u)$. From $\pi(u^*)<\pi(u)$ we know that when the outmost for-loop is at $u$, $u^*$ has already been proceeded so $\delta(v)\le d(v,u^*)< d(v,u)$. Consider when $v$ is added to $Q$, we have $\delta(v)>\delta(u')+c_{v,u'}=d(u,u')+c_{v,u'}\ge d(v,u)$, which contradicts to $\delta(v)<d(v,u)$.

\hide{The we prove that all vertices in $\chi_{\pi}^{(v)}$ dominate $v$. Assume to the contrary that at the end of the algorithm, $\exists u \in \chi_{\pi}^{(v)}$ and $u$ does not dominate $v$. This means there exists $u^*$ which satisfies both prior to $u$ and $d(v,u^*)<d(v,u)$. From $\pi(u^*)<\pi(u)$ we know that when the outmost for-loop is at $u$, $u^*$ has already been proceeded so $\delta(v)\le d(v,u*)< d(v,u)$. Consider when $v$ is added to $Q$, we have $\delta(v)>\delta(u')+c_{v,u'}=d(u,u')+c_{v,u'}\ge d(v,u)$, which contradicts to $\delta(v)<d(v,u)$.}

Lastly, since line 9 and line 10 are always executed together, once a vertex is added to the \seq~of another vertex, the corresponding pairwise distance is added to $d$, and the distance is guaranteed to be the shortest by the correctness of Dijkstra's algorithm.
\hide{If $u$ dominates $v$, we have proved that all $w_i$ will be extracted to update $w_{i+1}$ one by one, and finally $\delta(v)$ is updated by $\delta(w_{k-1})$. This means when $v$ is extracted, $\delta(v)$ has been set as the shortest distance from $v$ to $u$, and this is then stored in $d$. Thus the distance from a vertex to its dominating vertices are in $d$ at the end of the algorithm.}

Also, since $\delta(v)$ is always decreasing during the algorithm, the
dominant vertices are added to $\chi_{\pi}^{(v)}$ in decreasing order
of their distance to $v$.
\end{proof}

\section{Asymptotic Tight Bound of Ramsey partition in Theorem~\ref{thm:FRTRAM}}

In case of any unclear, here we show that in Theorem~\ref{thm:FRTRAM}, $\alpha=\left(1-2^{-1/2a}\right)=\Omega(\gamma)$.
\begin{proof}
Let $f(t)=2^{-t}+\frac{1}{2}t$, then $\frac{\mathrm{d}f}{\mathrm{d}t}=-2^{-t}\ln 2+\frac{1}{2}$.

For any $t$ that $0<t\leq\frac{1}{4}$, we have $-2^{-t}\ln 2+\frac{1}{2}<-2^{-1/4}\ln 2+\frac{1}{2}<0$, which shows that $f$ monotonously decreases in range $\left(0,\frac{1}{4}\right]$. Thus,
\begin{equation}
\label{fand1}
f(t)<f(0)=1
\end{equation}
for $0<t\le \frac{1}{4}$.

Any integer $a$ that $a\ge 2$ leads to $0<\frac{1}{2a}\le \frac{1}{4}$. Plugging $t=\frac{1}{2a}$ in (\ref{fand1}), we get $f(\frac{1}{2a})=2^{-1/2a}+\frac{1}{4a}<1$, which further shows that $1-2^{-1/2a}>\frac{1}{2}\cdot\frac{1}{2a}=\Omega(\frac{1}{2a})$. \end{proof}

\section{Approximation Factor for Theorem~\ref{thm:FRTRAM}}
\label{sec:approx}

For any fix integer $a>1$, we now try to analyze a pair of vertices $u$ and $v$ in graph $G$. Let $u$ be the first padded point in the pair in the $t$-th FRT tree, which means that $u$ is ``far away'' from the boundaries of the hierarchical partitioning in that tree.
Assume that the lowest common ancestor of the vertices be in the $i$-th level in the $t$-th FRT tree.
Since $u$ is padded, the distance $d_G(u,v)$ is at least $\left(1-2^{-1/2a}\right)\cdot{2^{-i-1}\beta\Delta}$.  The distance reported from the tree query $d_T(u,v)$ is no more than $2\cdot 2^{-i}\beta\Delta$.
Therefore, $d_G(u,v)\geq \left(1-2^{-1/2a}\right)/4\cdot d_T(u,v)$.

Storing this data structure requires expected $O\left({1/\epsilon}\cdot n^{1+1/(a(1-\epsilon))}\right)$ space.  Notice that any $a\leq 2$ does not make sense since storing the distances for the whole matric space need $O(n^2)$ space.
For any $a>2$, let $\epsilon=\lfloor a/3\rfloor/a$, so  distance oracle in Theorem~\ref{thm:FRTRAM} using $O\left(n^{1+1/k}\right)$ space (integer $k=a-\lfloor a/3\rfloor$) provided $18.5k$ distance approximation (achieved maximum $18.33k$ when $a=3$).

\end{document}